\documentclass[10pt,notitlepage,aps,pra,showpacs,superscriptaddress,singlecolumn,floatfix,footinbib,longbibliography]{revtex4-1}

\usepackage[colorlinks=true,citecolor=blue,urlcolor=magenta]{hyperref}
\usepackage[utf8]{inputenc}
\usepackage[T1]{fontenc}
\usepackage[sc,osf]{mathpazo}%\linespread{1.05}
\usepackage{amsmath, amsthm, amssymb,amsfonts}
\usepackage{graphicx}
\usepackage{dcolumn}
\usepackage{bm}
\usepackage{bbm}
\usepackage{mathtools}
\usepackage{comment}
\usepackage{color}

\DeclareMathOperator{\vrho}{\varrho}

\def\k{{\bf k}}
\def\vrho{\varrho}

\newcommand{\floor}[1]{\left \lfloor #1 \right \rfloor}

\newcommand{\bea}{\begin{eqnarray}}
\newcommand{\eea}{\end{eqnarray}}

\def\bi{\begin{itemize}}
\def\ei{\end{itemize}}
\def\bc{\begin{center}}
\def\ec{\end{center}}

\def\ket#1{|#1\rangle}
\newcommand{\one}{\mbox{$1 \hspace{-1.0mm}  {\bf l}$}}
\def\tr{\mathrm{tr}}
\def\ket#1{\left| #1\right>}
\def\bra#1{\left< #1\right|}

\newcommand{\kb}[2]{\ket{#1}\bra{#2}}
\newcommand{\braket}[2]{\left< #1|#2\right>}

\newcommand{\ketbra}[1]{| #1 \rangle\langle #1 |}
\DeclareMathOperator{\BB}{\mathcal{B}}

\newcommand{\II}{\mathcal I}
\newcommand{\GG}{\mathcal G}

\newcommand{\Sy}{\mathcal S}
\newcommand{\PP}{\mathcal P}

\newcommand{\td}{\widetilde}

\newtheorem{theorem}{Theorem}

\newtheorem{lemma}[theorem]{Lemma}

\newtheorem{observation}[theorem]{Observation}

\begin{document}
\title{Simulating extremal temporal correlations} 

\author{Cornelia Spee}
\affiliation{Institute for Quantum Optics and Quantum Information (IQOQI),
Austrian Academy of Sciences, Boltzmanngasse 3, 1090 Vienna, Austria}
\affiliation{Naturwissenschaftlich-Technische 
Fakult\"at, Universit\"at Siegen, Walter-Flex-Stra{\ss}e 3, 57068 Siegen, Germany}
\author{Costantino Budroni}
\affiliation{Faculty of Physics, University of Vienna, Boltzmanngasse 5, 1090 Vienna, Austria}
\affiliation{Institute for Quantum Optics and Quantum Information (IQOQI),
Austrian Academy of Sciences, Boltzmanngasse 3, 1090 Vienna, Austria}
\author{Otfried G\"uhne}
\affiliation{Naturwissenschaftlich-Technische 
Fakult\"at, Universit\"at Siegen, Walter-Flex-Stra{\ss}e 3, 57068 Siegen, Germany}

\begin{abstract}  
The correlations arising from sequential measurements on a single quantum 
system form a polytope. This is defined by the arrow-of-time (AoT) constraints, 
meaning that future choices of measurement settings cannot influence past outcomes.
We discuss the resources needed to simulate the extreme points of the AoT polytope, 
where resources are quantified in terms of the minimal dimension, or 
``internal memory'' of the physical system. First, we analyze the equivalence 
classes of the extreme points under symmetries. 
Second, we characterize the minimal dimension necessary to obtain a given extreme 
point of the AoT polytope, including a lower scaling bound 
in the asymptotic limit of long sequences.  
Finally, we present a general method to derive dimension-sensitive temporal 
inequalities for longer sequences, based on inequalities for shorter ones, 
and investigate their robustness to imperfections.
\end{abstract}
\maketitle

%%%%%%%%%%%%%%%%%%%%%%%%%%%%%%%%
\section{Introduction}
%%%%%%%%%%%%%%%%%%%%%%%%%%%%%%%%
The study of spatial correlations, from Bell nonlocality~\cite{Bell1964,nonloc_rev} 
to entanglement theory~\cite{ent_revH,ent_revG}, has had, on the one hand, a profound 
impact on the foundations of quantum mechanics. On the other hand, it stimulated plenty 
of applications in quantum information processing, such as quantum key 
distribution~\cite{AcinQKD2007}, randomness certification~\cite{Pironio2010} and expansion~\cite{Colbeck2011}, to mention a few. Moreover, correlations stronger 
than quantum ones, but still obeying the no-signaling constraints~\cite{PopescuFPH1994}, 
have been extensively investigated both from a fundamental perspective
and in relation with applications to quantum information processing.

Similarly, temporal correlations have been studied from the perspective of the 
difference between classical and quantum systems, mostly in the framework of 
Leggett-Garg inequalities~\cite{LeggettPRL1985,EmaryRPP2014} and noncontextuality 
inequalities \cite{KlyachkoPRL2008,CabelloPRL2008} tested via sequential measurements~\cite{Kirchmair2009,GuhnePRA2010}. More recently, a notion of 
non-classical temporal correlations has been formulated also from a different 
perspective that does not require assumptions on the {\it noninvasivity} or 
{\it compatibility} of the measurements \cite{BrierlyPRL2015,BudroniNJP2019}. 
Few quantum information processing tasks have been formulated directly in this 
framework, such as dimension witnesses~\cite{BudroniPRL2014,SchildPRA2015,Hoffmann2018,Spee2020}, 
purity certification~\cite{Spee2019}, and time-keeping devices~\cite{BVW2020}. Many other tasks, despite not 
being directly formulated in the language of temporal correlations, are closely 
related, since they naturally involve sequential operations. This is, for example, the 
case for prepare-and-measure scenarios~\cite{GallegoPRL2010,BrunnerPRL2013}, quantum random access codes (QRACs)~\cite{Wiesner1983,Ambainis1999, Ambainis2002, BowlesPRA2015, AguilarPRL2018, Miklin2019},   classical simulations of quantum
contextuality~\cite{KleinmannNJP2011,Fagundes2017}, quantum simulation of classical 
stochastic processes~\cite{GarnerNJP2017}, memory asymmetry between prediction and retrodiction~\cite{ThompsonPRX2018}, and optimal ticking clocks~\cite{Woods2018}.

The temporal counterpart to the no-signaling constraints~\cite{PopescuFPH1994} are the 
arrow-of-time (AoT) constraints~\cite{ClementePRL2016}, stating the impossibility of 
signaling from the future to the past. These conditions define the AoT polytope~\cite{ClementePRL2016}. 
It has been shown that the extreme points of this polytope are given by the deterministic assignments~\cite{HoffmannThesis2016,AbbottPRA2016,Hoffmann2018}, where each output in
a sequence is obtained as a deterministic function of the previous inputs and outputs. 
This implies that any such point is realizable by sequential measurements on a physical 
system, even for a classical theory, if the internal memory of the system is large enough 
to store the information about previous inputs and outputs. This is in stark contrast to 
the spatial case, where different correlations correspond to different theories, often 
irrespectively of the system dimension. A difference in temporal correlations between 
classical and quantum theory is recovered if the system 
dimension is constrained, as noted already a long time ago in the context of QRACs~\cite{Wiesner1983,Ambainis1999, Ambainis2002}, and similarly those theories can be distinguished from
generalized probability theories (GPTs)~\cite{BudroniNJP2019}. This dimension dependence can be 
exploited to construct temporal inequalities that can certify a lower bound on 
the dimension of the system, i.e. they are dimension
witnesses~\cite{GallegoPRL2010,BrunnerPRL2013,WolfPRL2009,GuehnePRA2014,Hoffmann2018,BudroniNJP2019,Spee2020, BowlesPRA2015}. 

In this paper, we investigate the minimal resources necessary to simulate a given 
extreme point of the AoT polytope. Here, the resource is quantified by the dimension
(or memory) needed to reproduce the outcomes of a measurement sequence. First, 
we study the symmetries of the AoT polytope w.r.t. classical post-processing. 
Then, we determine the minimal dimension required for the realization of a given 
extreme point. As in Refs.~\cite{Hoffmann2018,BudroniNJP2019,Spee2020}, no assumption on the concrete realization or quantum description of a measurement is made, it is only assumed that the same measurement may be carried out at different times. We then continue by providing a simple method to combine dimension-sensitive
temporal inequalities for shorter sequences to obtain inequalities valid for longer 
sequences. Finally, we discuss the robustness of temporal inequalities if the 
measurements are not perfect, i.e. if they vary over time. 

We remark that the simulation of extreme temporal correlations, for a given length and number of inputs and outputs, may or may not be sufficient to simulate {\it all} temporal correlations, depending on the assumption on the resources available. Consider the scenario in which the experimenter has two machines, able to generate the correlations $p_1$ and $p_2$, respectively. At the beginning of each measurement sequence she chooses with probability $\lambda$ the machine $1$ and with probability $1-\lambda$ the machine $2$, and then she proceeds to measure the whole sequences with the initially chosen machine. In this way, she would obtain as correlation $p$ the convex mixture $p=\lambda p_1 + (1-\lambda)p_2$. By straightforwardly extending this argument, one can see that the randomness available in the choice of the initial machine gives rise to all convex mixtures. In particular, by being able to simulate all extreme correlations of a given scenario (i.e., the extreme points of a given polytope), the experimenter can simulate all corresponding temporal correlations (i.e., the convex hull, corresponding to the full polytope). This is the scenario discussed in Ref.~\cite{BudroniNJP2019}. In contrast, if this initial randomness is not an available resource, the set of temporal correlations has a much more complicated structure~\cite{Mao2020}.

The paper is organized as follows. In Sect.~\ref{sec:prel}, we  introduce our notation and the considered scenario in detail and we  review some important previous results. In Sect.~\ref{sec:symm}, we further explore the properties of the AoT polytope, in particular the symmetries of the AoT polytope under classical  post-processing,  i.e., possible relabeling of the inputs and outputs, as such transformations do not affect the quantum realization.  In Sect.~\ref{sec:mindim}, we investigate questions such as how much  memory is required to realize a given extreme point within quantum theory, or stated differently, what is the minimal dimension that is necessary to obtain this correlation.  In Sect.~\ref{sec:lowerb}, we  provide a lower bound on the dimension needed to obtain an arbitrary extreme point and an estimate of the behavior in the asymptotic limit of arbitrary long measurement sequences, for any number of inputs and outputs.  Then, in Sect.~\ref{sec:combining}  we show a general method to 
produce temporal inequalities for longer sequences by combining inequalities for shorter ones. Finally in Sect.~\ref{sec:imperf},  we investigate how robust are our statements on temporal correlations in the case in which the assumption of repeated measurements is only approximately satisfied.

%%%%%%%%%%%%%%%%%%%%%%%%%%%%%%%%%%%%%%%%%%%%%%%%%%%
\section{Notation and Preliminaries}
\label{sec:prel}
%%%%%%%%%%%%%%%%%%%%%%%%%%%%%%%%%%%%%%%%%%%%%%%%%%%

\begin{figure}[t]
\includegraphics[width=0.5\columnwidth]{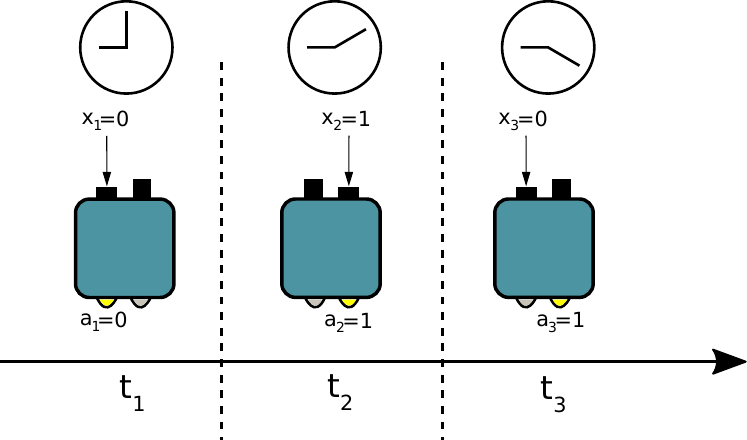} 
\caption{Finite-state machine: A single box is provided an input sequence $x_1 ,x_2, x_3$ and generates an output $a_1, a_2, a_3$ at different instants of time. No external clock/memory is accessible to the box and hence its behavior is solely governed by its internal state. Mathematically, this corresponds to having transformation rules for the internal state of the machine that are time-independent.
\label{fig:setup}}
\end{figure} 

We consider the scenario of sequential measurements depicted in Fig.~\ref{fig:setup}. A box receives a sequences of {\it inputs}, or {\it measurement settings}, $x_1,x_2,\ldots,x_L$ and produces a sequences of {\it outputs}, or {\it measurement outcomes} $a_1,a_2,\ldots,a_L$. The machines works by transforming probabilistically its internal state, e.g., a quantum state $\rho$, according to the measurement input and outcome and generating a measurement outcome according to the input and previous state. We are then interested in the correlations $p(a_1a_2\ldots a_L|x_1x_2\ldots x_L)$. 

More concretely, an operation on a quantum system, associated with an
input $x$, is represented by a quantum instrument, namely a collection of completely 
positive maps $\{\II_{a|x}\}_a$,  that sum up to a unital map, 
i.e., $\sum_a \II_{a|x}(\openone)= \openone$, where $\openone$ denotes the identity operator,
corresponding to the rule of preservation of probability in the Heisenberg picture, 
see, e.g., \cite{HeinosaariZiman2011} for a textbook introduction. 
Each instrument $\{\II_{a|x}\}_a$ defines a generalized measurement, i.e., a positive operator valued measure (POVM), 
through the formula $E_{a|x}:= \II_{a|x}(\openone)$. Correlations for a sequence of inputs $x_1, x_2$ and outputs $a_1, a_2$
are given by the formula
\begin{equation}\label{eq:def_qprob}
p(a_1a_2|x_1x_2)= \tr[\rho\ \II_{a_1|x_1} (E_{a_2|x_2}) ] = \tr[\rho\ \II_{a_1|x_1} \circ \II_{a_2|x_2}(\openone ) ],
\end{equation}
where $\circ$ denotes the composition of maps, and analogous expressions hold for longer sequences. We assume that the evolution of our box is time-independent, except for the external classical inputs that are provided at each time step. Practically, this assumption means two things. First, that the different correlations are generated by the transitions of the internal state of the machine; in quantum mechanical terms, this implies that for a given input $x$ the machine  applies the quantum instrument $\{\II_{a|x}\}_a$ independently of which time step $t$ we are in. This is already implicit in Eq.~\eqref{eq:def_qprob}, since we used only the symbol $\II_{a|x}$ to denote the instruments, without any reference to the time step, e.g., we want to calculate the probability $p(00|00)$ we apply the same mapping twice, i.e., $p(00|00)=\tr [\rho\ \II_{0|0} \circ \II_{0|0}(\openone ) ]$. Second, we assume that the inputs are provided at equally spaced time intervals and the free evolution of the system is always 
implemented by the same quantum channel, e.g., one can think about an evolution governed by a time-independent Hamiltonian. Hence, the time evolution can be reabsorbed wlog into the definition of the quantum instruments. Boxes satisfying these assumptions are called {\it finite-state machines} by generalizing a well-known classical notion~\cite{PazBook2003} (see also Ref.~\cite{BudroniNJP2019} for more details on the quantum and generalized probability theory case). We equivalently say that the measurement operations are {\it time-independent}.

In this scenario with time-ordered measurements, any theory that respects causality must satisfy the so-called {\it arrow of time} (AoT) conditions~\cite{ClementePRL2016}, namely, the future choice of inputs cannot modify the probabilities of past outcomes. For the simple case of a sequence of two measurements, the correlation $p(a_1a_2|x_1x_2)$ must satisfy
\begin{equation}
\sum_{a_2} p(a_1a_2|x_1x_2)= \sum_{a_2} p(a_1a_2|x_1x_2'), \text{ for all } a_1,x_1, x_2, x_2'.
\end{equation}
This condition is analogous to the no-signaling conditions for spatial correlations~\cite{PopescuFPH1994}, but it constrains only one direction, i.e., signaling from the future to the past. These linear constraints, together with positivity, $p(a_1a_2|x_1x_2)\geq 0$, and normalization, $\sum_{a_1a_2} p(a_1a_2|x_1x_2)=1$, define a polytope called the AoT polytope~\cite{ClementePRL2016}, denoted in the general case as $\PP_L^{O,S}$, where $O$ denotes the number of outputs, $S$ the number of inputs (or measurement settings) and $L$ the length of the measurement sequence. 

Such constraints are satisfied by classical and quantum mechanics and it has been proven that all extreme points are given by deterministic assignments, i.e. correlations which have the property that for any input one obtains a deterministic output~\cite{HoffmannThesis2016,AbbottPRA2016,Hoffmann2018}.  Intuitively, this comes from the fact that the AoT constraints allows us to decompose the probability distribution as
\begin{equation}
p(a_1a_2a_3|x_1x_2x_3)=p(a_1|x_1) p(a_2|a_1;x_1x_2) p(a_3|a_1a_2;x_1x_2x_3);
\end{equation}
 the extreme points, then, are given by the products of deterministic functions, generating $a_1, a_2$ and $a_3$ respectively, from the previous inputs and outputs. It has been shown in ~\cite{Hoffmann2018} that the AoT polytope $P_{L}^{O,S}$ has 
 \bea
N_{L}^{O,S}= (O^S)^{\frac{S^L-1}{S-1}}
 \eea extreme points. The extreme points can, then, be reached if the machine has enough ``internal memory'', namely, a large enough set of perfectly distinguishable internal states~\cite{Hoffmann2018}, to remember previous inputs and outputs and generate deterministically the corresponding outputs.

For a sequence of length $L$, an extreme point of the AoT polytope can be represented as a {\it tree graph} with $\sum_{k=0}^{L-1} S^{k}=\frac{S^L-1}{S-1}$ nodes, as depicted in Fig.~\ref{fig:tuple-tree}. The tree graph can be intuitively understood as follows. At each time-step the evolution of our system ``branches'' depending on the received input, e.g., the history in which the system received $0$ departs from the history in which it received $1$, since the internal state of the machine will evolve differently. In this way, we can keep track of all possible sequences of inputs that are obtained up to length $L$ and the corresponding evolution of the internal state. Moreover, since the strategy is deterministic, to each node of the graph corresponds a unique state. Of course, the same state may be used several times in the whole evolution of the system (provided that it deterministically generates the correct sequence of outputs, as we will see later in more detail). Hence, to each node, which we denote with the pair $(l,k)$ where  $l\in\{1,\ldots,L\}$ denotes the time-step, and $k\in\{1,\ldots,S^{l-1}\}$ denotes in which node of the time-step $l$ we are, we can associated a tuple $\Gamma_{l,k}=(z_1,z_2,\ldots,z_{S})$,  where $z_i\in\{0,\ldots,O-1\}$ denotes the (deterministically generated) outcome of the measurement $M_i$. Moreover, the state associated with the node $(l,k)$ encodes also the information on all the subsequent deterministically generated outputs, which motivates the following definition. A (sub)tree $T_{l,k}^r$, called the $r$-{\it length future} of $(l,k)$, is a collection of tuples connected to a root node $(l,k)$, representing the current and future deterministic outcomes. It is defined as  $T_{l,k}^r:= \{\Gamma_{l,k}, \Gamma_{l+1,h_1^{(
1)}},\ldots, \Gamma_{l+1,h_{S}^{(1)}},\Gamma_{l+2,h^{(2)}_1},\ldots, \Gamma_{l+2,h^{(2)}_{S^2}}, \ldots, \Gamma_{l+r,h^{(3)}_{1}}, \ldots,\Gamma_{l+r,h^{(r)}_{S^r}} \}$,
where  $h^{(m)}_l\in \{(k-1) S^{m}+1,\ldots,k S^m\}$;  we denote $T_{l,k}^{L-l}$ simply as $T_{l,k}$ and call it the {\it future} of $(l,k)$. See Fig. \ref{fig:tuple-tree} for more details. 

\begin{figure}[h]
\begin{center}
\includegraphics[width=0.95\columnwidth]{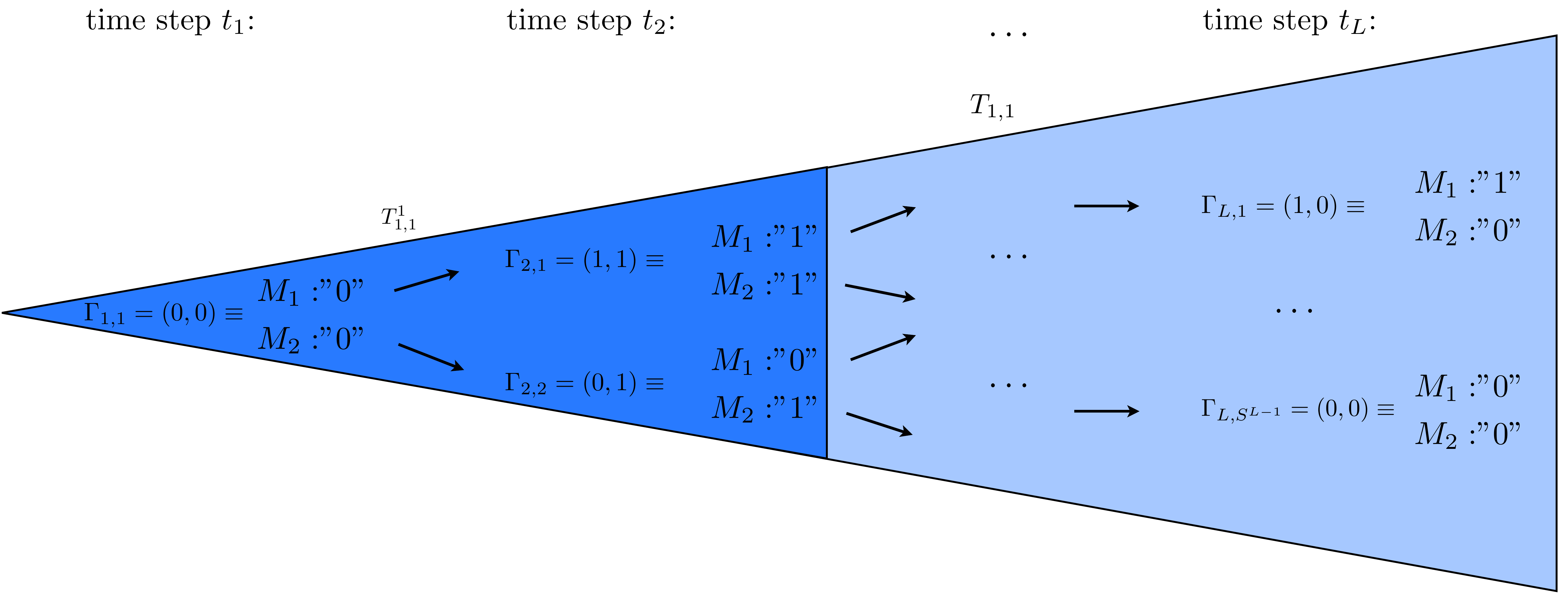} 
\end{center}
\caption{Assignment of tuples for an extreme point of $\PP_L^{2,2}$.  In  this example,  for both measurements the outcome ``$0$'' is obtained in the first time step. Then, then evolution branches, depending whether $M_1$ or $M_2$ has been measured, giving, the branches with $\Gamma_{2,1}$ or with $\Gamma_{2,2}$, respectively. After measuring $M_1$ in the first time step, one obtains in the second time step for both measurements deterministically the outcome ``$1$'' (corresponding to the tuple $\Gamma_{2,1}$) and after performing $M_2$ in the first time step one observes for measurement $M_1$ ($M_2$) in the second time step the outcome ``$0$'' (``$1$'') respectively (corresponding to the tuple $\Gamma_{2,2}$. The $1$-length future of $(1,1)$, i.e., $T_{1,1}^1$ is indicated by the dark blue triangle, and the history of the whole sequence, $T_{1,1}$ is represented by the triangle obtained by joining the dark blue and light blue regions.}
\label{fig:tuple-tree}
\end{figure}

To each node $(l,k)$  corresponds an internal state $\rho_{l,k}$ (quantum or classical) of the machine that generates deterministically the tuple of outcomes $\Gamma_{l,k}$. We remark that classical and quantum states are able to generate the same deterministic strategies, hence the distinction is at this point irrelevant (for more details see Sec. \ref{sec:mindim}). Moreover, since the procedure is deterministic, to the same state must correspond the same sequence of future outcomes, namely,
\begin{equation}\label{eq:equiv_state_future}
\rho_{l,k} = \rho_{l',k'} \Rightarrow T_{l,k}^r = T_{l',k'}^r, \text{ for } r=\min\{ L-l, L-l'\}.
\end{equation}
If two tuples $\Gamma_{l,k}$ and  $\Gamma_{l',k'}$ satisfy $T_{l,k}^r = T_{l',k'}^r$,  for $r=\min\{ L-l, L-l'\}$, we  say that they have {\it equivalent futures}, see Fig.~\ref{fig:inequiv}. Notice that the notion of equivalent futures depends always on the maximal length $L$ of the observed sequences, and that the comparison between subtrees makes sense only for equal lengths $r$ chosen as above. For a given deterministic sequence, we call $T_{1,1}$, i.e., the entire tree, the {\it history} of the sequence. 

This observation provides a way of counting the minimal number of states necessary for reproducing an extreme point of the AoT polytope, since inequivalent futures must correspond to different states. In particular, in order to check whether  $T_{l,k}^r \neq T_{l',k'}^r$, one can simply check for shorter sequences, i.e., whether $T_{l,k}^{s} \neq T_{l',k'}^{s}$ for $s=0,\ldots,r$.

\begin{figure}[t]\centering
\includegraphics[width=\columnwidth]{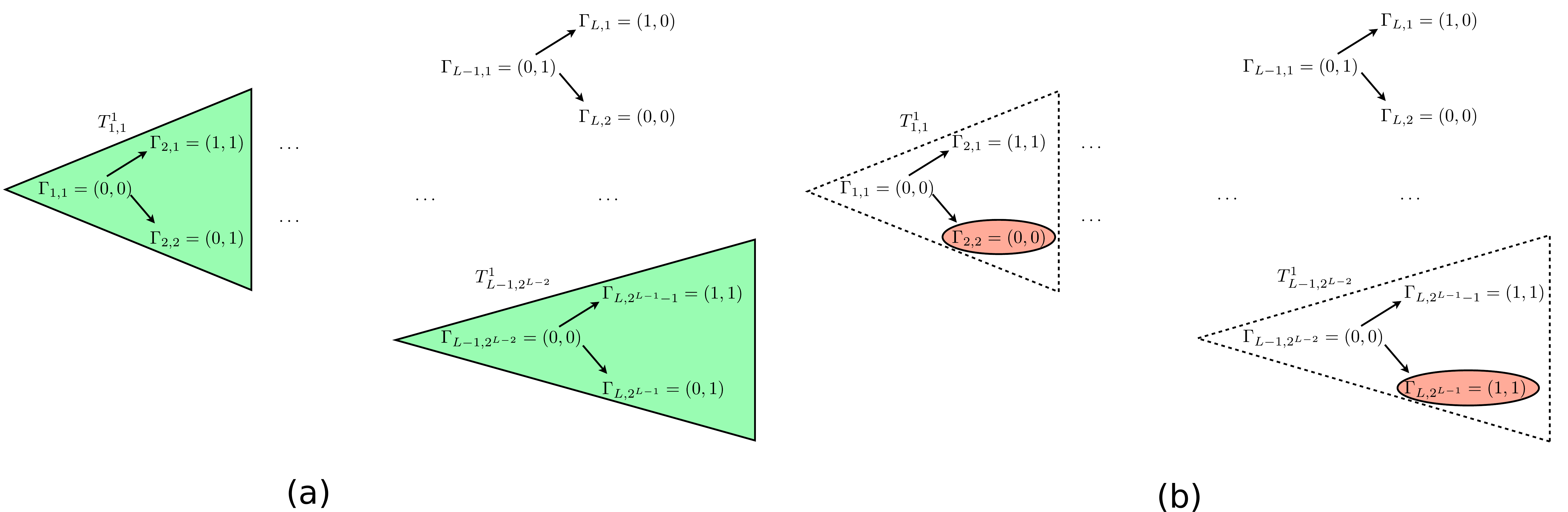}
\caption{We show two examples of, respectively, equivalent and inequivalent futures. In this case, since we want to compare the $(L-1)$-th node with another one, in this case the first, we need to consider only the $1$-length futures. In other words, in the comparison between tuples $\Gamma_{l,k}$ and  $\Gamma_{(L-1),k'}$, we need to consider $T_{l,k}^r = T_{L-1,k'}^r$,  with $r=\min\{ L-l, 1\}$. (a) An example of two equivalent futures. The $1$-length futures of $(1,1)$ and $(L-1,2^{L-2})$ are equivalent, since the nodes themselves, as well as the tuples occurring in the next time-step after the measurements leading to  them,  coincide.  (b) Example of two inequivalent futures. The $1$-length futures of $(1,1)$ and $(L-1,2^{L-2})$ are inequivalent as the respective subtrees do not coincide.\label{fig:inequiv} } 
\end{figure}

\section{Symmetries of the AoT polytope}\label{sec:symm}
It has been shown that, if no assumption on the dimension of the quantum system is made, any correlation in the polytope can be realized. The ability of realizing a correlation is independent of the chosen labeling of the outcomes and/or measurement settings as long as one performs exactly the same relabeling at every time step. This is due to the fact that any such relabeling can be implemented classically even after the measurement sequences have been performed, i.e. such relabelings correspond to some classical post-processing. Note that the condition that the same relabeling is applied to all time steps is necessary to be consistent with our assumption of time-independent measurements. In the following, we characterize the number of equivalence classes of extreme points under these symmetries for small numbers of settings and outcomes.

In particular, we  define an {\it outcome relabeling equivalence} (ORE) class 
as an equivalence class of extreme points  w.r.t. to the relation of being the same up to a relabeling of the outcomes. In particular, since relabeling is a classical post-processing, if one extreme point is obtainable by measurements on a physical system, the same is true for all elements in the class. Of course, also the measurement settings can be subject to relabeling and all extreme points that are equivalent up to relabeling of the measurement settings can be realized within the same physical implementation. Then, we define the  {\it relabeling equivalence} (RE) classes of extreme points as the set of extreme points that are equal up to relabeling of outcomes and measurement settings.

\subsection{General considerations}
In the temporal scenario the only relevant symmetries are given by the relabeling of inputs and outputs of a given sequence. The corresponding symmetry groups are given by the symmetric groups $\Sy_{O}$ and $\Sy_S$, where as defined before $O$ and $S$ are the number of outputs and inputs, respectively. The total group of symmetry is given by the direct product $\GG:=\Sy_O \times \Sy_S$.

To each element $g\in \GG$, we  associate a transformation $S_g$ on the extreme points of the polytope $\PP_L^{O,S}$. In this way, each extreme point $v$ of  $\PP_L^{O,S}$  generates an orbit $O_v$ defined as
\begin{equation}
O_v = \{  S_g v \ | \ g \in \GG\}
\end{equation}
The action of a group on a set naturally induces an equivalence relations in terms of orbits given by $v\sim w \Leftrightarrow O_v = O_w$. In this case, belonging to the same orbit means that the extreme point $v$ can be obtained from the extreme point $w$ via a relabeling of inputs and outputs, and vice versa.

The number of equivalence classes is then given by the number of different orbits. Hence, if one can evaluate the number of orbits one can deduce the number of RE classes for a given scenario. Below we show how to do so by identifying the elements that are invariant under a symmetry, their orbits and the cardinality of the orbits, which allows us to deduce the number of orbits, for the case of two outcomes and two and three settings. The same procedure can be applied to arbitrary number of outcomes and settings without extra conceptual difficulties;  however, as the symmetric group grows, recall that there are $n!$ permutations of $n$ elements, the whole procedure becomes much longer and tedious.

\subsection{Relabeling of outcomes and measurement settings}
For the case $O=2$ one can define for each ORE class  as a representative an extreme point having the property that for any measurement the outcome in the first time step is ``$0$''. This allows us to count the number of ORE classes as given in the following lemma.
\begin{lemma}\label{lemma:ORE}
The number of ORE classes of extreme points of $P_{L}^{2,S}$ is given by $N_{ORE}^{2,S,L}=(2^S)^{\frac{S^L-S}{S-1}}$.
\end{lemma}
\begin{proof}
We consider one representative of each ORE class and show that the cardinality of the ORE class is $2^S$. The number of equivalence classes is the total number of extreme points divided by the cardinality of such classes. It is straightforward to see that one particular choice of a representative is given by demanding the outcome of each measurement setting at the first time step to be ``$0$''. Note that, obviously, any relabeling of the outcomes must alter the outcome at the first time step. Note further that all subsequent outcomes specify the ORE class and any sequence of these outcomes is possible.

For each measurement setting there are $2$ possible outcomes in the first time step and therefore there are $2^S$ possible relabelings of the outcomes at the first time step. This implies that the cardinality of a ORE class is given by $2^S$.
Using that the number of extreme point of $P_{L}^{2,S}$ is given by $(2^S)^{\frac{S^L-1}{S-1}}$ (see Ref.~\cite{Hoffmann2018}), the number of ORE classes is then
\begin{equation}
\frac{(2^S)^{\frac{S^L-1}{S-1}}}{2^S}=(2^S)^{\frac{S^L-S}{S-1}},
\end{equation}
which concludes the proof.
\end{proof}

For the relabeling of the measurement settings, we  first present a counting argument for $O=S=2$. Then, we  extend our investigation to the case of $O=2$ and $S=3$.

\begin{lemma}
The number of RE classes of extreme points of $P_{L}^{2,2}$ is given by $\frac{1}{2}( 4^{2^L-2} + 4^{(2^{L-1}-1)})$.
\end{lemma}

\begin{proof}
 \begin{figure}[b]
\includegraphics[width=0.6\columnwidth]{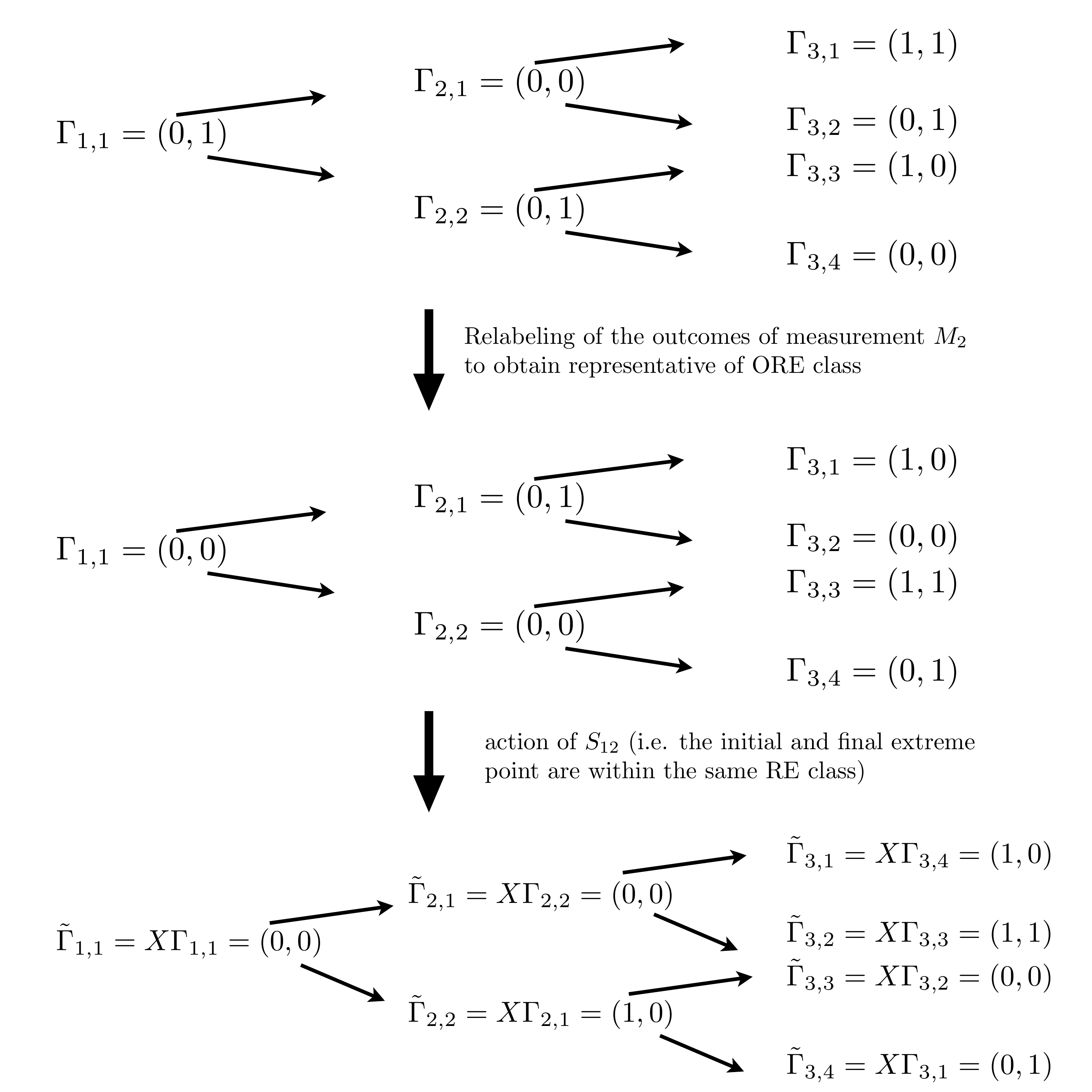} 
\caption{Illustration of the action of the group on an example.  First we transform the extreme point to the representative of the ORE class, which has the property that in the first time step all outcomes are „0“ by relabeling the outcomes accordingly. Then we illustrate the action of the symmetry $S_{12}$ (which permutes the measurement settings) on this representative, where the tuples obtained by applying $S_{12}$ are denoted by $\tilde{\Gamma}_{k,l}$ and the operator $X$ permutes the the elements of the tuples $\Gamma_{k’,l’}$ of the extreme point on which the transformation is performed. The other group element of $\Sy_2$, i.e., the identity element $e$, leaves any extreme point invariant.
\label{fig:symmetrygroup}}
\end{figure} 

As we know already from Lemma~\ref{lemma:ORE} the number of equivalence classes under the relabeling of outcomes, it is sufficient to study the action of the permutation of the inputs on the representatives of the ORE classes. Since $S=2$, we have the group $\Sy_2=\{e, S_{12}\}$, where $e$ is the identity element and $S_{12}$ exchange the first and second input (see also Fig. \ref{fig:symmetrygroup}). In particular $S_{12}^2=e$.
For a given length $L$, the number of equivalence classes $N^{(L)}$ is given  by 
\begin{equation}
N^{(L)} = N_{\rm I}^{(L)} + N_{\rm N}^{(L)},
\end{equation}
where $N_{\rm I}^{(L)}$ is the number of orbits consisting of only one vector, i.e., the vectors invariant under the action of the group (see Fig. \ref{fig6}), and $N_{\rm N}^{(L)}$ the number of orbits consisting of two vectors, i.e., vectors not invariant under the action of the group. To count the number of invariant vectors, we apply the following argument. First, let us fix the outcome relabeling by choosing the first input as ``$0$''. Given that the sequence is invariant under exchange of inputs until step $m-1$, the possible ways of completing it in an invariant way at the step $m$ are $4^{2^{m-2}}$ out of $4^{2^{m-1}}$ possible completions. In fact, for each fixed choice of inputs and outputs at the step $m-1$ there are $4$ possible ways of completing the last two outcomes for the two settings. Out of the $2^{m-1}$ possible choices of inputs, we need to fix only half of them, because the other half is fixed by the input exchange symmetry $S_{12}$, see Fig.~\ref{fig6}. 

\begin{figure}[t]
\begin{center}
\includegraphics[width=0.7\columnwidth]{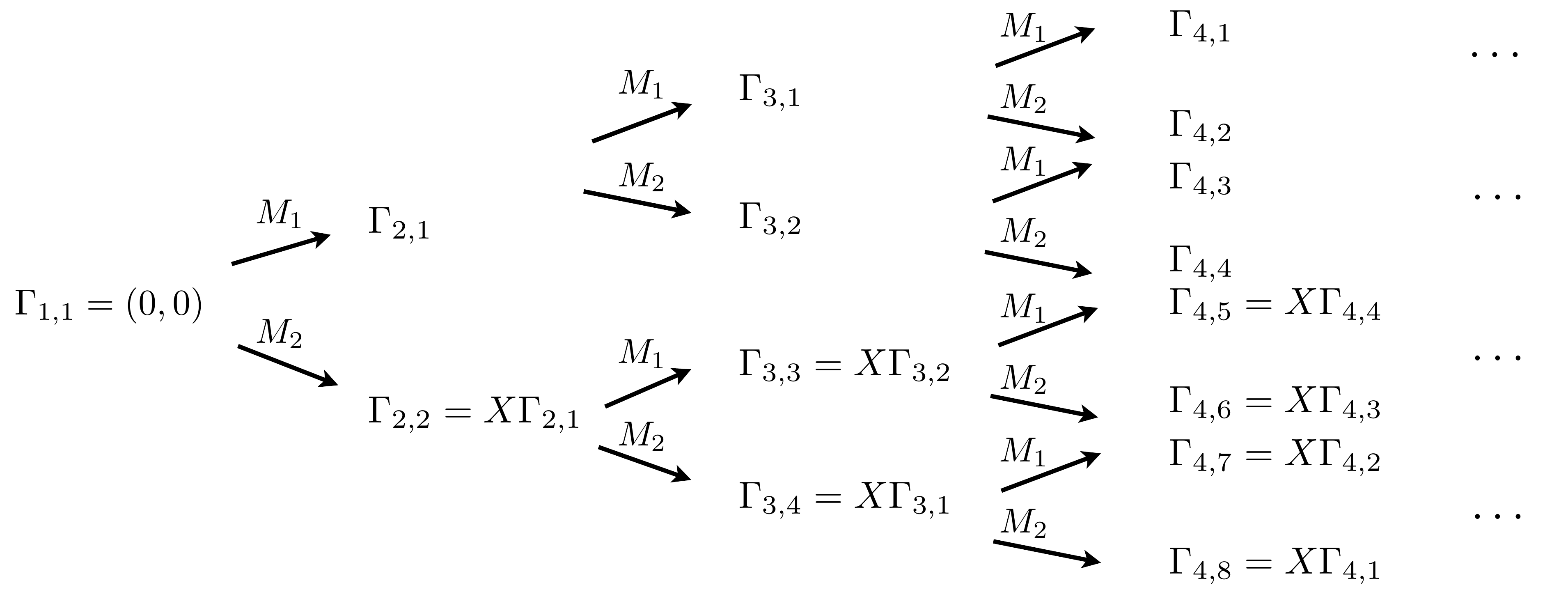} 
\end{center}
\caption{ Equivalent tuples for an extreme point invariant under settings relabeling, in the $O=S=2$ scenario. The operator $X$ permutes the the  elements of the tuple. E.g., The tuple $\Gamma_{2,2}$ corresponds to the outcomes for $M_1$ and $M_2$ in the second step, after $M_2$ has been measured. For the extreme point to be invariant under relabeling of outcomes, it must be that $\Gamma_{2,2}=X \Gamma_{2,1}$, i.e., the outcomes are, up to relabeling, those that would have been obtained had we measured $M_1$ in the first step.}
\label{fig6}
\end{figure} 

The number of possibilities to extend a vector at the step $m$, given that it is symmetric at the step $m-1$, is thus $4^{2^{m-2}}$. We can then compute the number of invariant vectors up to length $L$ as
\begin{equation}
N_{\rm I}^{(L)}:= \prod_{i=2}^L 4^{2^{m-2}} =4^{2^{L-1}-1}.
\end{equation}

Moreover, we can compute the number of equivalence classes of non-invariant vectors as 
\begin{equation}
N_{\rm N}^{(L)} =  \frac{1}{2}(N_{\rm ORE}^{(L)} - N_{\rm I}^{(L)}),
\end{equation}
where $(N_{\rm ORE}^{(L)} - N_{\rm I}^{(L)})$ is the number of non-invariant vectors, and the factor $1/2$ comes from the fact that each orbit contains two elements.

Finally, we can write
\begin{equation}
N^{(L)} = N_{\rm I}^{(L)} + N_{\rm N}^{(L)} =  \frac{1}{2}(N_{\rm ORE}^{(L)} + N_{\rm I}^{(L)})=\frac{1}{2}( 4^{2^L-2} + 4^{2^{L-1}-1}).
\end{equation} 
\end{proof}

For the most simple scenario, $O=S=L=2$, this implies that there are ten RE classes. As already discussed in \cite{Hoffmann2018}, six of these classes can be obtained with a qubit, whereas for four of these classes a qutrit is required. In Table \ref{table1}, we provide a representative for each of these classes and indicate whether a qubit or a qutrit is necessary in order to realize a member of this class. In the following section, we  then present a general theorem which allows us to deduce from a given extreme point (with arbitrary $O, S$ and $L$) the dimension that is necessary and sufficient to realize it.

\begin{table}
  \begin{tabular}{ | l | c | r |} \hline
      Extreme point [with $\Gamma_0=(0,0)$] & minimal dimension \\ \hline
       $\Gamma_1=(0,0), \Gamma_2=(0,0)$ &  1 \\ \hline
   $\Gamma_1=(0,0), \Gamma_2=(1,1)$ &   2\\ \hline
  $ \Gamma_1=(0,0), \Gamma_2=(0,1)$ &   2\\ \hline
   $\Gamma_1=(0,0), \Gamma_2=(1,0)$ &   2\\ \hline
  $ \Gamma_1=(0,1), \Gamma_2=(0,1) $&   2\\ \hline
   $\Gamma_1=(1,1), \Gamma_2=(1,1) $&   2\\ \hline
   $\Gamma_1=(0,1),\Gamma_2=(1,0) $&   3\\ \hline
       $ \Gamma_1=(1,0),\Gamma_2=(0,1) $&   3\\ \hline
 $\Gamma_1=(0,1),\Gamma_2=(1,1)$ &   3\\ \hline
$\Gamma_1=(1,1),\Gamma_2=(0,1)$ &   3\\ \hline
\end{tabular}
\caption{This table shows a representative for each of the 10 RE classes for $O=S=L=2$ and the minimal dimension which allows to reach a member of the class (see also Theorem \ref{Theoreml2}). Note that  for $O=S=L=2$ the RE classes and their minimal dimension has been already identified in \cite{HoffmannThesis2016} and a corresponding table can be also found there (with a different choice of representatives).}
\label{table1}
\end{table}

After having gained an understanding of the case of two inputs, we generalize our approach for counting the RE classes to the case of three inputs. 

\begin{lemma}\label{LemmanumRE23}
The number of RE classes of extreme points of $P_{L}^{2,3}$ is given by $2^{{\frac{L-1}{2}+\frac{3(3^L-3)}{4}}-1} + \frac{1}{6}[ 2^{3\frac{3^L-3}{2}} + 2^{\frac{3^L-3}{2}+1} ].$
\end{lemma}
In order to prove the Lemma, we  start again from the equivalence classes of outcome relabeling and impose only the conditions for the relabeling of the inputs. In this case, we have the permutation group of three elements, $\Sy_3$, representing the permutation of the inputs. The group $\Sy_3$ consists of the following elements
\begin{equation}
\begin{split}
&e\\
&S_{12}\\
&S_{23}\\
&S_{13}=S_{23} S_{12} S_{23} = S_{12} S_{23} S_{12}\\
&\sigma_{123} = S_{12} S_{23} = S_{23} S_{13}\\
&\sigma_{132} = S_{13} S_{23} = S_{23} S_{12}
\end{split}
\end{equation}
We therefore write the total number of equivalence classes as
\begin{equation}
N^{(L)} = N_{\rm I}^{(L)} + N_{S}^{(L)} + N_{\sigma}^{(L)} + N_{\rm N}^{(L)},
\end{equation}
where $N_{\rm I}^{(L)} , N_{\rm N}^{(L)} $ are defined as above, as the orbits of vectors that are invariant or non-invariant under any symmetry respectively, and $N_{S}^{(L)} $ ($N_{\sigma}^{(L)} $) are the orbits of vectors invariant under only one of the $S_{ij}$ (vectors invariant only under  $\sigma_{123}$ or $\sigma_{132}$) respectively. Note here that $S_{ij}^2=e$ and that $\sigma_{123} \sigma_{132} = \sigma_{132} \sigma_{123}=e$. In Appendix \ref{AppnumRE23} we count the number of invariant and non-invariant orbits and with this prove Lemma \ref{LemmanumRE23}.

\section{Minimal dimension for given extreme points}\label{sec:mindim}
It is a well known result in quantum state discrimination that two states have orthogonal ranges, corresponding to a trace-distance of $1$, if and only if they can be perfectly discriminated, i.e., with probability $1$, by a single measurement (cf., e.g., Ref.~\cite{NC_book} Ch. 9). More precisely, this fact can be stated as follows
\begin{observation}\label{obs:orth_range}
Let $E$ be an effect of a POVM and $\rho_1=\sum_{i\in I} p_i\kb{\Psi_i}{\Psi_i}$ with $p_i>0$ ($\rho_2=\sum_{k\in K} q_k\kb{\Phi_k}{\Phi_k}$ with $q_k>0$) the spectral decomposition of a density matrix $\rho_1$ ($\rho_2$) respectively. Then $\tr \{\rho_1 E\}=1$ and $\tr \{\rho_2 E\}=0$ only if $\braket{\Psi_i}{\Phi_k}=0$ for all $ i\in I$ and $k\in K$. 
\end{observation}

It is important to notice that via a single POVM $E$ one can represent not only a single measurement, but also a sequence, e.g., ${E_{abc|xyz}:=\II_{a|x}\circ \II_{b|y} \circ \II_{c|z} (\openone)}$, where the maps $\{\II_{a|x}\}_a$ represent the quantum instrument in the Heisenberg picture. This implies that not only states that produce a different outcome with probability one are orthogonal, but also states that produce a different sequence of outcomes with probability one are orthogonal.

Using this, we are able to determine the minimum dimension that is required for a quantum system to obtain a given extreme point of $P_2^{O,S}$.

\begin{theorem}\label{Theoreml2}
Given an extreme point $p$ of $P_2^{O,S}$, the minimal dimension $d$ needed to obtain it is given by the number of inequivalent tuples in the history of $p$, i.e., $T_{1,1} =\{\Gamma_{1,1},\Gamma_{2,1},\ldots ,\Gamma_{2,S}\}$. In particular, a system with dimension $d=S+1$ can always reach all extreme points of $P_2^{O,S}$, independently of the number of outcomes, as this is the maximal number of tuples in $T_{1,1}$.
\end{theorem}
\begin{proof} According to Eq.~\eqref{eq:equiv_state_future}, to different futures correspond different states. In this particular case, namely, $L=2$, we need to compare different $T_{l,k}^1$, i.e., single tuples. By Obs.~\ref{obs:orth_range}, such states must have orthogonal ranges. These two conditions already provide the minimal number of orthogonal states necessary to reach a given extreme point of the AoT polytope. Intuitively, orthogonality is the only relevant property for obtaining different futures, hence a minimal realization requires only pure states. This is confirmed by the explicit construction below, which uses only pure states. Given the tuples $\Gamma_{1,1},\Gamma_{2,1},\ldots ,\Gamma_{2,S}$,  there may be repetitions, which in this simple case of $L=2$ corresponds to having equivalent futures. We can rewrite them as $d$ tuples $\{\Gamma_1,\ldots,\Gamma_d\} =\{\Gamma_{1,1},\Gamma_{2,1},\ldots ,\Gamma_{2,S}\}$ with inequivalent futures, i.e., $\Gamma_{i}\neq \Gamma_j$. We associate to each 
of them a vector $\ket{k}$ from the ONB $\{ \ket{k}\}_{k=1}^d$. Without loss of generality, we can assume that $\Gamma_{1,1}=\Gamma_1=(0,\ldots,0)$, i.e., we fix all the measurement outcomes at the first step to be zero. This simply means that we relabel the outcome of all measurements such that $0$ is obtained for all of them on the initial state. Then, we fix the initial state as $\rho_{\rm in} = \ketbra{1}$. The measurements are constructed as follows
\begin{equation}\label{eq:povm_const_gamma}
E_{a|x} = \sum_{j \in J_{a|x}} \ketbra{j}, \text{ with } J_{a|x} := \{ j \in \{1,\ldots,d\}\ |\ [\Gamma_j]_x = a\},
\end{equation}
for $a=0,\ldots,O-1$, $x=1,\ldots,S$. Clearly, $E_{a|x}\geq 0$ and $\sum_{a} E_{a|x} = \openone$, for all $x$, so they are valid POVMs. In particular, $E_{0|x}=\ketbra{1}+ \sum_{j \in J_{0|x},j\neq 1} \ketbra{j}$ for all $x$. The corresponding Kraus operators $\{K_{0|x}^j\}_{j\in J_{0|x}}$ providing the postmeasurement state, i.e., in the Schr\"odinger picture $\rho \mapsto \sum_j K_{0|x}^j \rho K_{0|x}^{j\dagger}$, are of the form $K_{0|x}^1=\ket{s}\bra{1}$ if $\Gamma_s = \Gamma_{2,x}$ or a tuple with equivalent future and $K_{0|x}^j=\ketbra{j}$ for $j\in J_{0|x}$ and $j\neq 1$. By construction, there are at most $S+1$ tuples, hence this number provides an upper bound on the minimal dimension necessary to reach any extreme point of $P_2^{O,S}$.
\end{proof}

We now discuss a specific example in $P_2^{2,3}$ to illustrate how Theorem \ref{Theoreml2} can be applied in order to determine the dimension which is necessary and sufficient to realize a given extreme point. For this we  consider the extreme point given by $\Gamma_{1,1}=(0,0,0)$, $\Gamma_{2,1}=(0,0,0)$,  $\Gamma_{2,2}=(1,1,1)$ and $\Gamma_{2,3}=(0,0,1)$. There are three inequivalent tuples in the history given by $\Gamma_{1,1}, \Gamma_{2,2}$ and $\Gamma_{2,3}$, as $\Gamma_{1,1}=\Gamma_{2,1}$. Hence, we have that any system which can realize this extreme point has at least dimension three. The Kraus operators of the measurements $x\in\{1, 2, 3\}$ for outcome $0$, which allow us to obtain this extreme point from the initial state $\rho_{\rm in} = \ketbra{1}$, can be chosen as $K_{0|x}^1=\ket{x}\bra{1}$ for all $x$ and $K_{0|x}^2=\ketbra{3}$ for $x=1,2$ (for $x=3$ there is only a single Kraus operator). Note that for the outcome $1$ it is sufficient to know the effects (as the post-measurement state does not need to be specified), which directly follow from the ones for outcome $0$. That is, we have $\mathcal{E}_{1,x}=\ketbra{2}$ for $x=1,2$ and $\mathcal{E}_{1,3}=\ketbra{2}+\ketbra{3}$. 

The same argument used to derive Theorem \ref{Theoreml2} can be generalized to sequences of arbitrary length.

\begin{theorem}\label{theomindimgen}
The minimal dimension $d$ required to reach an extreme point $p$ of $P_L^{O,S}$, is given by the number of inequivalent futures $T_{l,k}^r$ in the history $T_{1,1}$. 
\end{theorem}
\begin{proof}
The proof generalizes straightforwardly from the case $L=2$ above. Again, different futures must correspond to different states, such states must be orthogonal, and can be chosen to be pure, providing a minimal-dimension representation.

The explicit construction of the model can then be extended from the previous one. Let us assign the state $\ket{1}$ as initial state, i.e., to $T_{1,1}$. Then compare $T_{1,1}$ with $T_{2,k}$, for $k=1,\ldots,S$, if they are equivalent, assign the same state $\ket{1}$ to $T_{2,k}$, otherwise, assign a new orthogonal state $\ket{2},\ket{3},\ldots$ to the future $T_{2,k}$. Repeat the same operation for $T_{3,k'}$, $k'=1,\ldots ,S^2$, assigning new state for any future inequivalent to $T_{1,1}$ or any $T_{2,k}$. Repeat again until the end of the tree, i.e., $T_{L,k}=\Gamma_{L,k}$, $k=1,\ldots, S^{L-1}$. To each node $\Gamma_{l,k}$ of the tree $T_{1,1}$ is then assigned a pure state $\ket{l,k}\in \{ \ket{j}\}_{j=1}^d$, possibly with repetitions. As in Th.~\ref{Theoreml2}, POVMs elements are constructed as projectors providing the correct outcomes for each state $\ket{j}$, as in Eq.~\eqref{eq:povm_const_gamma}. Similarly, Kraus operators $\{ K_{a|x}^j\}_j$, associated with the POVM element $E_{a|x} = \sum_{j \in 
J_{a|x}} \ketbra{j}$, consist in measure-and-prepare operations $K_{a|x}^j=\ket{i}\bra{j}$, when the state $\ket{j}$ emits the output $a$ for the measurement $x$ and transition to the state $\ket{i}$, all with probability one. For the last time-step, i.e., from $L-1$ to $L$, one can use diagonal Kraus operators analogously to the construction in Th.~\ref{Theoreml2}.
\end{proof}

Note that the above protocol does not involve any coherences, as all states and effects are diagonal in the same basis and the state-update rule also involves transitions within the same basis, hence, it can be realized with a classical system.

\section{Lower bound on the dimension which is necessary to realize any extreme point}\label{sec:lowerb}
In the following, we provide a construction of extreme points for any $L$ from which one can determine a lower bound on the minimal dimension required for its realization. This lower bound is then automatically also a lower bound on the dimension necessary to realize any extreme point. The main result is that the minimal dimension scales, at least, exponentially in $L$. 
Let us consider the polytope $P_L^{O,S}$. 
The main idea of the proof can be briefly explained as follows. Consider a sequence of length $L$ and take a time-step $j< L$. If the number of remaining time-steps $L-j$ is big enough, for the tuples $\{\Gamma_{j,s}\}_s$ we can choose their futures $\{T_{j,s}\}_s$ to be all different, hence, each $\Gamma_{j,s}$ will be associated to an orthogonal state and the number of such tuples will provide a lower bound on the minimal dimension necessary for their realization. Our argument consists in estimating the maximum $j$ such that this is possible.

 For the history $T_{1,1}$ at time-step $j$, there exist $S^{j-1}$ different subtrees $T_{j,k}$. If $j$ is properly chosen, namely, $L-j$ is large enough such that we can construct different futures $T_{j,k}$ for $k=1,\ldots, S^{j-1}$, then the realization of such an extreme point requires at least $d=S^{j-1}$. First notice that the number of possible futures of length $x$ is given by $(O^S)^{\frac{S^{x+1}-1}{S-1}}= O^{\frac{S^{x+2}-S}{S-1}}$. 
Hence, $j$ must be selected in such a way that the remaining sequence allows us to assign different futures (which are of length $L-j$) to each node, namely as the largest integer such that
\begin{equation}
S^{j-1} \leq O^{\frac{S^{L-j+2}-S}{S-1}}.
\end{equation}
We can further simplify the expression using the identity $O^x=S^{x \log_S O}$
\begin{equation}
\begin{split}
S^{j-1}\leq S^{\frac{S^{l-j+2}-S}{S-1}\log_S O} \Leftrightarrow j-1 \leq \frac{S^{l-(j-2)}-S}{S-1}\log_S O  = S^{-(j-1)} \left(\frac{S^l}{S-2} \log_S O\right) - \frac{S}{S-1}\log_S O.
\end{split}
\end{equation}
This equation can be solved in terms of the principal branch of Lambert function $W$, namely, the function implicitly defined as the solution to the equation $x e^x = k$, i.e., $x e^x = k \Leftrightarrow x=W(k)$. In this case, let us see how to solve it for the equation $x = a S^{-x} + b$
\begin{equation}
\begin{split}
x = a S^{-x} + b \Leftrightarrow (x-b)S^{x-b} = a S^{-b} \Leftrightarrow  (x-b)e^{(x-b) \ln S} \ln S= a S^{-b}\ln S,\\
 \text{We have } y e^y = a S^{-b}\ln S, \text{ for } y:= (x-b)\ln S \Leftrightarrow y=W(a S^{-b} \ln S) \Leftrightarrow x = \frac{W(a S^{-b} \ln S)}{\ln S} + b.
\end{split}
\end{equation}
By substituting $x=j-2$, $a= \left(\frac{S^L}{S-1} \log_S O\right)$, $b=- \frac{S}{S-1}\log_S O$ one obtains the condition 
\begin{equation}\label{eq:ineq_j_W}
j\leq \frac{W\left[\left(\frac{S^L}{S-1} \log_S O\right) S^{\frac{S}{S-1}\log_S O} \ln S\right]}{\ln S} - \frac{S}{S-1}\log_S O + 2,
\end{equation}
which gives the maximal $j$ as
\begin{equation}
j = \floor{ \frac{W\left[\left(\frac{S^L}{S-1} \log_S O\right) S^{\frac{S}{S-1}\log_S O} \ln S\right]}{\ln S} - \frac{S}{S-1}\log_S O + 2},
\end{equation}
where $\floor{x}$ denotes the floor of $x$, i.e., the largest integer smaller than $x$.

To compute the asymptotic scaling, one can write
\bea
S^{j-2}\leq O^{-\frac{S}{S-1}}  e^{W[\frac{S^L  O^{\frac{S}{S-1} } \ln S \log_S O}{S-1}]}.
\eea
Using that $\ln (x)- \ln [\ln (x)]+ \frac{\ln [\ln (x)]}{2 \ln (x)}\leq W(x)$ for $x\geq e$ \cite{Hoorfar2008} we can obtain a lower bound on the minimal dimension as follows. For $m\in \mathbb{R}$ such that
\begin{equation}
\begin{split}
S^{m-2}\leq \frac{ O^{-\frac{2 S}{S-1}}S^{L} \ln O }{S-1} \left(\ln \left[\frac{S^L  O^{\frac{S}{S-1} } \ln O }{S-1}\right]\right)^{-1+\left(2 \ln \left[\frac{S^L  O^{\frac{S}{S-1} } \ln O }{S-1}\right]\right)^{-1} },
\end{split}
\end{equation} 
the minimal dimension satisfies
\begin{equation}
d_{\min } \geq S^{m-2},
\end{equation}
where the ``$-2$'' term takes into account the fact that $m$ may not be an integer. For large $L$ such a lower bound scales as
\begin{equation}
A S^{L-1}  (\ln[B S^L])^{-1+(2\ln[B S^L])^{-1}}\approx {\alpha} e^{\beta L+ (\frac{\delta}{L}- \gamma) \ln L },
\end{equation}
for appropriately chosen constants $A,B,\alpha,\beta,\gamma,\delta$. This proves that the minimal dimension required to reach any extreme point scales at least exponentially (up to logarithmic corrections). In Appendix~\ref{sec:appA}, we present a different construction of an extreme point which can provide an improved lower bound on the minimal dimension, however no closed formula on the scaling.

\section{Combining temporal inequalities}\label{sec:combining}

In the following, we present a method for deriving new inequalities for temporal correlations for sequences of length $nL$ with $n\in \mathbb{N}^+$, based on the knowledge of inequalities for the shorter length $L$.  
It instructive to first describe the method for a simple example, based on the inequalities for the case $L=O=S=2$ derived in Ref.~\cite{Hoffmann2018}. The original inequalities were derived by computing the qubit bound for expressions of the form $\sum_{x,y=0,1} p(a_{x} b_{xy}|x,y)$, for a specific choice of the outputs $\{a_x, b_{xy}\}_{x,y}$ where the the algebraic bound $4$ is achieved by an extreme point of AoT polytope $P_2^{2,2}$, i.e., $p(a_x b_{xy}|xy)=1$ for all $x,y$. Up to symmetries, four expression were derived, namely
\begin{align}
\BB_{1} := p(00|00) + p(00|11) + p(01|01) + p(01|10)\leq C_1,
\nonumber
\\
\BB_2 := p(01|00) + p(01|11) + p(00|01) + p(00|10)\leq C_2,
\nonumber
\\
\BB_3 := p(01|00) + p(00|11) + p(01|01) + p(01|10) \leq C_3,
\nonumber
\\
\BB_4 := p(01|00) + p(01|11) + p(01|01) + p(00|10) \leq C_4.
\end{align}
each one corresponding to one of the extreme points of the AoT polytope $P^{2,2}_2$, which cannot be reached by qubit strategy, namely,
\begin{align}
e_{1} : p(00|00) = p(00|11) = p(01|01) = p(01|10) =1 , \text{ and } 0 \text{ otherwise;}
\nonumber
\\
e_2 : p(01|00) = p(01|11) = p(00|01) = p(00|10) =1 , \text{ and } 0 \text{ otherwise;}
\nonumber
\\
e_3 : p(01|00) = p(00|11) = p(01|01) = p(01|10)  =1 , \text{ and } 0 \text{ otherwise;}
\nonumber
\\
e_4 : p(01|00) = p(01|11) = p(01|01) = p(00|10) =1 , \text{ and } 0 \text{ otherwise.}
\end{align}
In general, to each extreme point $e_i$, with components labelled by $\vec{a}=(a_1,\ldots,a_L),\vec{x}=(x_1,\ldots,x_L)$, i.e., ${[e_i]_{\vec{a},\vec{x}}=p(\vec{a}|\vec{x})}$, we can associate a temporal inequality
\begin{equation}\label{eq:def_Bi}
\BB_i = \sum_{\vec{a},\vec{x} } c_{\vec{a},\vec{x}}^{(i)} p(\vec{a}|\vec{x})\leq C_i, 
\end{equation}
where $c_{\vec{a},\vec{x}} := [e_i]_{\vec{a},\vec{x}}$, $c_{\vec{a},\vec{x}}\in \{0,1\}$ since $e_i$ is a deterministic strategy, and $C_i$ is the bound for a given dimension $d$ of the quantum system, corresponding to the algebraic bound $\sum_{\vec{a},\vec{x}} c_{\vec{a},\vec{x}}$, if the extreme point $e_i$ can be reached in dimension $d$. Note that the same bounds also hold for any other element within the same RE class as $e_i$.

Given two deterministic strategies for length two, we can construct a strategy for length four simply by combining them in the following way:
\begin{equation}\label{eq:eiprod}
p(a_1,a_2,a_3,a_4|x_1,x_2,x_3,x_4)=p(a_1,a_2|x_1,x_2) p(a_3,a_4|a_1a_2;x_1,x_2,x_3,x_4)
\end{equation}
where $p(a_1,a_2|x_1,x_2)$ and each $p(a_3,a_4|a_1a_2;x_1,x_2,x_3,x_4)$, is a deterministic strategy associated with an extreme point $e_k\in \{e_1,e_2,e_3,e_4\}$, implying that
\begin{equation}\label{eq:4step_factor}
p(\td a_1[x_1],\td a_2[x_1,x_2] \ |\ x_1,x_2) p(\td a_3[x_1,x_2,x_3],\td a_4[x_1,x_2,x_3,x_4]\ |a_1a_2;x_1,x_2,x_3,x_4)=1,
\end{equation}
for properly chosen functions $\{\td a_i\}$ of the inputs $\{x_i\}$, and $0$ otherwise.

We denote the corresponding extreme point of $P_{4}^{2,2}$ as $e_{\bf k} = (e_{k_1^1},e_{k_2^1},\ldots, e_{k_2^4})$ with ${\bf k}=(k_1^1, k_2^1,\ldots, k_2^4)$,  where $k_1^1$ labels the extreme point used for the first time-period , i.e., the two steps $t_1$ and $t_2$, and $k_2^i$ for $i=1,\ldots,4$ denote the possible extreme points for the second time period, i.e., steps $t_3$ and $t_4$, each belonging to one of the different branches of the evolution depending on the inputs $x_1,x_2$, as shown in Eq.~\eqref{eq:4step_factor}.  We, then, construct the associated inequality
\begin{equation}
\BB_{\bf k}= \sum_{\vec{a},\vec{x}} c_{\vec{a},\vec{x}} p(\vec{a}|\vec{x}) \leq  C_{k_1} C_{k_2}.
\end{equation}
where $C_{k_j}=\underset{i}\max\, C_{k_j^i}$, is the maximum taken for a given time period over all possible branches.
The proof of the bound is straightforward
\begin{equation}
\begin{split}
\BB_{\bf k}&= \sum_{\vec{a},\vec{x}} c_{\vec{a},\vec{x}} p(\vec{a}|\vec{x}) = \sum_{x_1,x_2,x_3,x_4} p(\td a_1[x_1],\td a_2[x_1,x_2],\td a_3[x_1,x_2,x_3],\td a_4[x_1,x_2,x_3,x_4]\ |x_1,x_2,x_3,x_4) \\
 &=  \sum_{x_1,x_2} p(\td a_1[x_1],\td a_2[x_1,x_2]\ |x_1,x_2) \sum_{x_3,x_4} p(\td a_3[x_1,x_2,x_3],\td a_4[x_1,x_2,x_3,x_4]\ |\td a_1[x_1],\td a_2[x_1,x_2]; x_1,x_2,x_3,x_4)\\
 &\leq \sum_{x_1,x_2} p(\td a_1[x_1],\td a_2[x_1,x_2]\ |x_1,x_2) C_{k_2} \leq C_{k_1} C_{k_2},
\end{split}
\end{equation}
where we used the AoT condition to break the probability and $C_{k_2}$ as an upper bound to the expression  $\sum_{x_3,x_4} p(\td a_3[x_1,x_2,x_3],\td a_4[x_1,x_2,x_3,x_4]\ |\td a_1[x_1],\td a_2[x_1,x_2]; x_1,x_2,x_3,x_4)$ for any value of $x_1,x_2,\td a_1[x_1],\td a_2[x_1,x_2]$, and finally the bound $C_{k_1}$. 

It is obvious that the above result depends only on the way of choosing a deterministic strategy, i.e., an extreme point of the AoT polytope, as a product strategy as in Eq.~\eqref{eq:eiprod}, the way of constructing the corresponding expression $\BB_{\bf k}$, and the knowledge of the bounds $\{C_{k_j^i}\}_{i,j}$ for the single expressions $\{\BB_{k_j^i}\}_{i,j}$ for ${\bf k}=(k_1^1,\ldots,k_n^m)$. We can then generalize the result as follows.

\begin{theorem}\label{Th8}
Given a collection of temporal inequalities involving $O$ outcomes, $S$ settings, and length $L$,
\begin{equation}
\BB_i = \sum_{\vec{a},\vec{x} } c_{\vec{a},\vec{x}}^{(i)} p(\vec{a}|\vec{x})\leq C_i, \text{ for } i=1,\ldots,N,
\end{equation}
associated to an extreme point $e_i$ and valid for quantum systems of dimension $d$, then, the following inequality for sequences of length $nL$ 
 \begin{equation}
 \BB_{\bf k} := \sum_{\vec{x}} p(\td a_1[x_1],\td a_2[x_1,x_2], \ldots \td a_{nL}[x_1,\ldots,  x_{nL}] |x_1,\ldots, x_{nL}) \\
\leq \prod_{j=1}^n C_j,
\end{equation}
with ${\bf k}=(k_1^1,\ldots,k_n^{m})$, $e_{k_j^i}\in \{e_1, \ldots, e_N\}$ and $C_{j}=\underset{i}\max\, C_{k_j^i}$
 also holds for quantum systems of the same dimension.
\end{theorem}
The proof of this theorem is analogous to the case of four time steps presented above. 
As an explicit example for this construction consider the extreme point $e_{\bf k} = (e_i,\ldots ,e_i)$ of length $2n$ with $i\in\{1,2,3,4\}$ and $n\in \mathbb{N}^+$ and the corresponding inequality $\BB_{\bf k}$. Then, according to Theorem \ref{Th8}, it holds that for qubits $\BB_{\bf k}\leq (C_i)^n$. Using Theorem \ref{theomindimgen}, it can be easily seen that  with a three-dimensional system one can reach the algebraic maximum of $\BB_{\bf k}=4^n$. It follows that the ratio of separation between a qubit and a qutrit is exponentially decreasing with the length of the sequence, i.e.  $(C_i/4)^n$. 

\section{Imperfect Implementation of Time-Independent Measurements}\label{sec:imperf}

The results obtained so far assume that the measurements are time-independent, i.e. the same input indicates that also the same measurement is implemented. Here, we  discuss how a deviation from this assumption influences our results. Before proceeding further, it is helpful to remark what we mean by imperfect implementation. What does it mean to ``perform the same measurement twice''? Consider the basic example of an apparatus that measures the spin of a particle either along the $X$ direction or the $Z$ direction with probability $1/2$ each. Clearly, in each round of the experiment when a sequence of two measurements is performed there is 50\% chance that two different measurements are performed, i.e., $X,Z$ or $Z,X$. However, according to our definition of time-independent quantum instruments, this situation is still allowed, since such an uncertainty is already contained in the definition of quantum instrument. The notion of imperfect implementation, hence, does not deal with random fluctuations in the 
measurement apparatus, but rather with some time-dependent drift in the parameters describing the measurement apparatus, e.g., a drift in the magnetic field orientation in the spin example. Notice, however, that Markovian time-evolutions can be still be absorbed in the definition of quantum instruments with a proper choice of measurement times.
 
In the following, we quantify the effect of such imperfect implementations of quantum instruments on the observed correlations. Such deviations can be quantified in terms of the diamond norm \cite{Kitaev1997}. It is also important to remark that in the following, it is more convenient to use the Schr\"odinger picture for the representation of quantum instruments.
This corresponds to take the dual $\II^*$ of the instruments appearing in Eq.~\eqref{eq:def_qprob}, acting now on states rather than observables. To avoid a heavy notation, however, we drop the superscript $^*$ in the remaining part of this section.

If $\II_{a|x}$ is the desired CP map for input $x$ and outcome $a$ and $\tilde{\II}_{a|x}$ is the one that is instead implemented in the experiment then $||\II_{a|x}-\tilde{\II}_{a|x}||_{\diamond}\leq \epsilon$ for all $x$ and $a$, where the diamond norm of a CP map $\II$ is defined as $\| \II \|_{\diamond}:= \max_{\rho_{AB} } \| \II_A\otimes {\rm id}_B (\rho_{AB}) \|_{\rm tr}$. Note that from the definition of the diamond norm it straightforwardly follows that $\tr[\tilde{\II}_{a|x}(\rho)-\II_{a|x}(\rho)]\leq ||\II_{a|x}-\tilde{\II}_{a|x}||_{\diamond}$ for all density matrices $\rho$. As we will see, this allows us to derive bounds on the influence of such a deviation on quantities that are linear in $p(ab\ldots |xy\ldots)$. In order to illustrate the basic idea we  consider here first two time steps and then three time steps, however, it is straightforward to generalize the bound to an arbitrary number of time steps. In particular, we obtain that
\begin{align}\nonumber
p(ab|xy)=&\tr \{\tilde{\II}_{b|y}[\II_{a|x}(\vrho_{\rm in})]\}\\\nonumber
=&\tr \{(\tilde{\II}_{b|y}-\II_{b|y})[\II_{a|x}(\vrho_{\rm in})]\}+\tr \{\II_{b|y}[\II_{a|x}(\vrho_{\rm in})]\}\\
\leq& \tr \{\II_{b|y}[\II_{a|x}(\vrho_{\rm in})]\}+\epsilon\,\tr [\II_{a|x}(\vrho_{\rm in})]
\end{align}
 and 
\begin{align}\nonumber
p(abc|xyz)=&\tr (\tilde{\II}_{c|z}\{\tilde{\II}_{b|y}[\II_{a|x}(\vrho_{\rm in})]\})\\\nonumber
=&\tr [(\tilde{\II}_{c|z}-\II_{c|z})\{\tilde{\II}_{b|y}[\II_{a|x}(\vrho_{\rm in})]\}]+\tr (\II_{c|z}\{\tilde{\II}_{b|y}[\II_{a|x}(\vrho_{\rm in})]\})\\ \nonumber
\leq& \tr (\II_{c|z}\{\tilde{\II}_{b|y}[\II_{a|x}(\vrho_{\rm in})]\})+\epsilon\, \tr [\II_{a|x}(\vrho_{\rm in})]\\\nonumber
=&\tr (\II_{c|z}\{(\tilde{\II}_{b|y}-\II_{b|y})[\II_{a|x}(\vrho_{\rm in})]\})+\tr (\II_{c|z}\{\II_{b|y}[\II_{a|x}(\vrho_{\rm in})]\}+\epsilon \,\tr [\II_{a|x}(\vrho_{\rm in})]\\
\leq &\tr (\II_{c|z}\{\II_{b|y}[\II_{a|x}(\vrho_{\rm in})]\}+2 \epsilon \,\tr [\II_{a|x}(\vrho_{\rm in})].
\end{align}
Note that we used here multiple times that $\tr[\tilde{\II}_{a|x}(\rho)-\II_{a|x}(\rho)]\leq ||\II_{a|x}-\tilde{\II}_{a|x}||_{\diamond}\leq\epsilon$ with $\tr (\rho)=1$, that $\tr (\tilde{\II}_{c|z}\{\tilde{\II}_{b|y}[\II_{a|x}(\vrho_{\rm in})]\})= \tr [\II_{a|x}(\vrho_{\rm in})]\tr \{\tilde{\II}_{c|z}[\tilde{\II}_{b|y}(\tilde{\rho}_{a|x})]\}$, where $\tilde{\rho}_{a|x}$ denotes the normalized post-measurement state for measurement $x$ and outcome $a$, and that 
$\tr (\II_{c|z}\{(\tilde{\II}_{b|y}-\II_{b|y})[\II_{a|x}(\vrho_{\rm in})]\})\leq  \| \II_{c|z}\{\tilde{\II}_{b|y}-\II_{b|y}\}\|_{\diamond} \tr[ \II_{a|x}(\vrho_{\rm in})] \leq \|\tilde{\II}_{b|y}-\II_{b|y}\|_{\diamond} \tr[ \II_{a|x}(\vrho_{\rm in})]\leq \varepsilon\tr[ \II_{a|x}(\vrho_{\rm in})]$, which follows from the contractivity of the trace distance under completely positive trace-nonincreasing operations. From this argument it can be easily seen that the deviation of the  probability $p(ab\ldots |xy\ldots)$  for $l$ time steps due to imperfect time-independent measurements can be bounded by $(l-1)\epsilon \,\tr [\II_{a|x}(\vrho_{\rm in})]$. Hence, temporal inequalities are still able to provide a reliable lower bound on the dimension.

As a final remark, it is interesting to notice the following. The above argument assumes certain quantum properties of the operations involved, hence, at least some partial characterization of the experimental devices. However, we simply noticed that instruments that are ``close'' in the quantum mechanical sense (and arguably the diamond norm is the natural distance among them) give rise to probability distributions that are again ``close'', with an error that scales linearly in the measurement length. 
Assumptions on such a distance, even if based on quantum mechanics,
do not necessarily require a full characterization of the experimental devices. It would be interesting to estimate the diamond norm in a device independent way, then our result helps in the design of improved experimental tests of temporal quantum correlations that rely on minimal assumptions and do not require a complete characterization of the measurement devices.

%%%%%%%%%%%%%%%%%%%%%%%%%%%%%%%%%%%%%%%%%%%%%%%%%%%
\section{Conclusion and Outlook}\label{sec:concl}
%%%%%%%%%%%%%%%%%%%%%%%%%%%%%%%%%%%%%%%%%%%%%%%%%%%
In this work, we studied the resources required to realize AoT correlations 
within quantum mechanics. We first identified which extreme points of the 
AoT polytope can be obtained by using the same  protocol followed by some 
classical post-processing of the input and output for the case of a small 
number of two-outcome measurements. Then we provided for an arbitrary given 
extreme point the dimension that is necessary and sufficient to realize it. 
In particular, we showed that this is given by the number of inequivalent 
futures  in the history associated with a point and we gave 
an explicit protocol that allows one to obtain it. We observed that this 
protocol does not involve any coherences and hence can be also implemented 
with a classical system. Moreover, we derived a lower bound on the minimal 
dimension that is necessary to reach an arbitrary extreme point for a given 
number of settings $S$, outcomes $O$ and time steps $L$ and we showed that 
in the asymptotic limit of long sequences this scales as $e^{\alpha L}/L$ 
(with $\alpha$ being some constant that depends on $O$ and $S$). 

In a previous work~\cite{Hoffmann2018}, extreme points of the AoT polytope 
have been used to construct dimension witnesses for sequences of short length. 
Here, we provided a general method to use these witnesses as building blocks 
for the construction of dimension witnesses for sequences of arbitrary length. 
Despite the fact that the bound on the so obtained temporal inequality is not 
necessarily tight, one finds inequalities which show an exponential scaling with respect to the 
length of the sequence. Finally, we made quantitative statements on how 
the bounds on linear temporal inequalities are affected if the assumption 
that at any time step one is able to implement the same measurement is 
violated. We showed that small deviations from these assumptions still 
allow us to deduce lower bounds on the dimension. 

There are several possible directions for future research. First, one can consider
a general point in the correlation polytope and consider the resources needed for
a simulation. This problem is challenging for two reasons: First, the quantum 
realizations for arbitrary points are difficult to find, and may require a larger 
dimension than the extremal points \cite{Mao2020}. Second, for a general point in the 
polytope a deterministic protocol is not suitable, so more general concepts, such as 
hidden Markov models \cite{Rabiner1986} or, more specifically, $\varepsilon$-transducers 
\cite{Barnett2015, Cabello2018} may be useful.

A second interesting problem comes from the observation that our simulation protocols
were purely classical, in the sense that they can be implemented using quantum states
diagonal in the computational basis. It would be
interesting to develop a general theory of temporal correlations, for 
which the quantum mechanical simulation requires less resources than the 
classical one, due to
effects like coherence  \cite{Gu2012}. This may open
a further way to test quantum devices using temporal correlations.

\section{Acknowledgements}
We thank Jannik Hoffmann, Yuanyuan Mao, and Zhen-Peng Xu for useful discussions. This work has been supported by the Austrian Science Fund (FWF) through the projects J 4258-N27 (Erwin-Schrödinger Programm), Y879-N27 (START project), ZK 3 (Zukunftskolleg), and F7113-N48 (BeyondC), by the Austrian Academy of Sciences and by the ERC (Consolidator Grant 683107/TempoQ).

%%%%%%%%%%%%%%%%%%%%%%%%%%%%%%%%%%%%%%%%%%%%%%%%%%

%%%%%%%%%%%%%      APPENDIX     %%%%%%%%%%%%%%%%%%%%%%%%%%%

%%%%%%%%%%%%%%%%%%%%%%%%%%%%%%%%%%%%%%%%%%%%%%%%%%%

\appendix

\section{Proof of Lemma \ref{LemmanumRE23}\label{AppnumRE23}}
In the following we prove Lemma \ref{LemmanumRE23}, i.e. we show that the number of RE classes of extreme points of the polytope $P_{L}^{2,3}$ is given by 
\bea 2^{{\frac{L-1}{2}+\frac{3(3^L-3)}{4}}-1} + \frac{1}{6}[ 2^{3\frac{3^L-3}{2}} + 2^{\frac{3^L-3}{2}+1} ].\eea

\begin{proof}
As mentioned in the main text, we  consider the ORE classes and impose then the conditions for the relabeling of the inputs. 
Hence, the relevant symmetry group has the following elements \begin{equation}
\begin{split}&e\\
&S_{12}\\
&S_{23}\\
&S_{13}=S_{23} S_{12} S_{23} = S_{12} S_{23} S_{12}\\
&\sigma_{123} = S_{12} S_{23} = S_{23} S_{13}\\
&\sigma_{132} = S_{13} S_{23} = S_{23} S_{12}.
\end{split}
\end{equation}
The total number of equivalence classes can  be written as
\begin{equation}
N^{(L)} = N_{\rm I}^{(L)} + N_{S}^{(L)} + N_{\sigma}^{(L)} + N_{\rm N}^{(L)},
\end{equation}
where $N_{\rm I}^{(L)} , N_{\rm N}^{(L)} $ are the number of orbits of vectors that are invariant or non-invariant under the whole symmetry group respectively, and $N_{S}^{(L)} $ ($N_{\sigma}^{(L)} $) count the orbits of vectors invariant under only one of the $S_{ij}$ (vectors invariant only under  $\sigma_{123}$ or $\sigma_{132}$) respectively. Notice that $S_{ij}^2=e$ and that $\sigma_{123} \sigma_{132} = \sigma_{132} \sigma_{123}=e$. 

Let us start with the set of invariant vectors. For $S=3$ the possible independent sequences of settings, i.e., sequences that are not generated one from the other by exchanging some settings, are given by:
\begin{equation}
l=1: 
\begin{array}{c} 
X
\end{array}; \qquad
l=2: \begin{array}{c} 
XX\\ XY
\end{array};\qquad
l=3: \begin{array}{c} 
XXX\\ XXY\\ XYX\\ XYY\\ XYZ
\end{array};
\end{equation}
It is clear that, since $S=3$, at each step the sequence of all identical measurements  generate only two new sequences, i.e., $XXX\ldots X$ and $XXX\ldots XY$, whereas all the other  generate three new ones. We can then count the number of such sequences as $Q_1=1$ and $Q_m=3(Q_{m-1}-1)+2= 3 Q_{m-1} -1$, giving 
\begin{equation}
Q_m=\frac{1}{2}(1+3^{m-1}).
\end{equation}
At each step $m$, then, we  need to choose $2^{Q_m}$ possible values, i.e., two values for each extra measurement setting added, and we do not count the step $m=1$, since this is fixed by the outcome relabeling symmetry. We then have
\begin{equation}
N_{\rm I}^{(L)}=\prod_{i=2}^{L} 2^{Q_i}=2^{\sum_{i=2}^L \frac{1}{2}(1+3^{i-1})}= 2^{\frac{1}{2}(L-1) + \frac{3^L-3}{4}}.
\end{equation}

For calculating $N_S^{(L)}$, we can first observe that if a vector $v$ is invariant under the action of $S_{12}$, i.e., $S_{12}v=v$, then the corresponding orbit is given by $O_v=\{ v_{12}, v_{23}, v_{13}\}$, where $v_{ij}$ is a vector invariant under the action of $S_{ij}$. It is sufficient to look at the case of $S_{12}$. If $S_{12}v = v$, we define $v_{12} := v$. Then, by action of the group we obtain $v_{23}:= S_{13} v$ and $v_{13}:=S_{23} v$. It can be straightforwardly verified that $v_{23}$ is invariant under the action of $S_{23}$. In fact, $S_{23} v_{23} = S_{23} S_{13} v = S_{23} S_{23} S_{12} S_{23} v = S_{12} S_{23} v = S_{12} S_{23} S_{12} v= S_{13} v = v_{23}$. A similar argument shows that $v_{13}$ is invariant under $S_{13}$. Hence, each orbit of a vector invariant under exactly on $S_{ij}$ generate the other two invariant vectors. We now need to count the number of representatives of such orbits.
\begin{equation}\label{eq:frac_s}
N_{S}^{(L)}= \td{N}_{S}^{(L)} - N_{\rm I}^{(L)}
\end{equation}
where $N_{S}^{(L)}$ is the number of representative vectors invariant under exactly one $S_{ij}$, but not fully invariant, $\td{N}_{S}^{(L)}$ is the number of vectors invariant under at least one of the $S_{ij}$, but possibly fully invariant. If the vectors are invariant under two $S_{ij}$, they are completely invariant, hence they are counted in $N_{\rm I}^{(L)}$. Note that the number of vectors invariant under exactly one $S_{ij}$ is given by $3 N_{S}^{(L)}$, since each equivalence class contains three vectors.

For each sequence, say $M_1M_1M_2M_3M_1M_2$, if we fix the outcomes of the sequences and then swap $M_1\leftrightarrow M_2$, we obtain another sequence. This applies to all of them, except the sequence $M_3M_3M_3M_3\ldots M_3$ which remains invariant. This implies that, assuming to be invariant up to step $m-1$, at step $m$ we  have $(2^3)^{\frac{3^{m-1}-1}{2}} 2^2$ choices of new outcomes that keep the sequence invariant under relabeling of settings (see also Fig. \ref{fig8}). Namely,  the number of tuples for which one can assign in time step $m$ one out of $(2^3)$ possibilities is divided by two, except for $M_3M_3M_3M_3\ldots M_3$, leading to $(2^3)^{\frac{3^{m-1}-1}{2}}$ ways of extending these sequences from time step $m-1$ to time step $m$. For the case $M_3M_3M_3M_3\ldots M_3(M_1,M_2,M_3)$ there are $2^2$ possible assignments, giving in total $(2^3)^{\frac{3^{m-1}-1}{2}} 2^2=(2^3)^{\frac{3^{m-1}-1}{2}+\frac{2}{3}}$. 

\begin{figure}[h]
\begin{center}
\includegraphics[height=0.4\columnwidth]{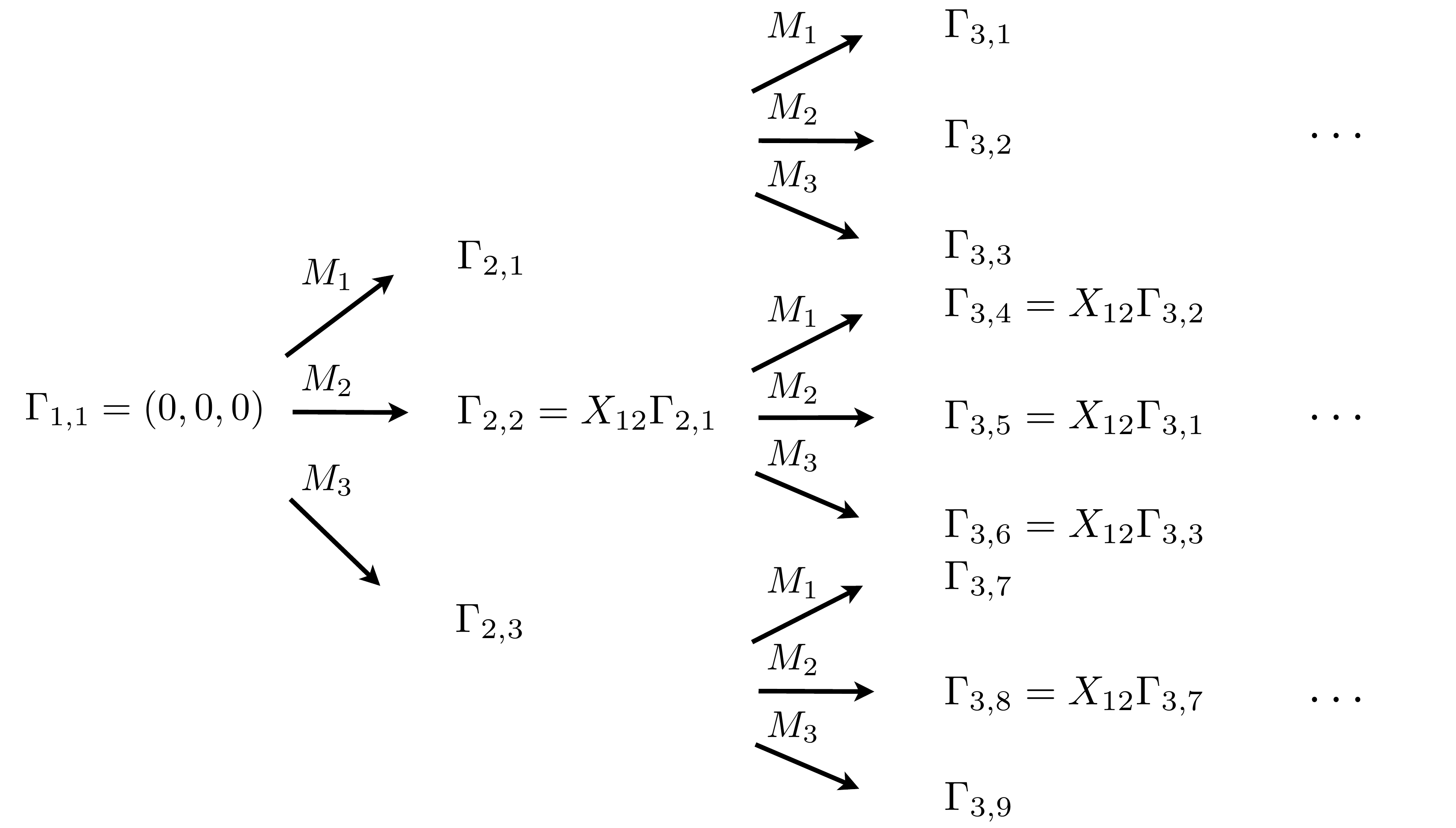} 
\end{center}
\caption{In this figure the tuples that have to be equal for an extreme point that is invariant under $S_{12}$ in the scenario of three measurement settings with each two outcomes are indicated. Here $X_{12}$ is an operator that permutes the outcomes of the measurements $M_1$ and $M_2$ in the tuples.}
\label{fig8}
\end{figure} 

Again, we do not count the step $m=1$ since it is fixed by outcome relabeling symmetry. We can then compute
\begin{equation}
\td{N}_S^{(L)} = \prod_{i=2}^L (2^3)^{\frac{3^{m-1}-1}{2}+\frac{2}{3}}= 2^{{\frac{L-1}{2}+\frac{3(3^L-3)}{4}}} 
\end{equation}

To compute $N_\sigma^{(L)}$ we need to consider the orbits of vectors invariant under $\sigma_{123}$. First notice that $\sigma_{123} v = v \Leftrightarrow \sigma_{132} v = v$ since $\sigma_{123}\sigma_{132}= \sigma_{132}\sigma_{123}=e$. We now prove that orbits are given by either $O_v=\{v, S_{ij}v\}$ for vectors invariant only under $\sigma_{ijk}$ and $O_v= \{v\}$ for vectors invariant under $\sigma_{ijk}$ and one of the $S_{kl}$. In fact, $\sigma_{123} v = v$ implies $S_{12} v = S_{12} S_{12} S_{23} v = S_{23} v = S_{23} S_{12} S_{23} v= S_{13} v$. Hence, if $S_{ij} v = v$ for any $ij$, then $v$ is invariant under the action of the whole group, i.e., $O_v=\{ v \}$. Otherwise, we obtain the orbit $O_v= \{ v, S_{ij} v\}$, where $S_{ij} v$ is the same vector for $ij=12,23,13$. It is important to notice, however, that if $v$ is invariant, i.e., $\sigma_{123}v=v$, then also $S_{ij}v$ is invariant, i.e., $\sigma_{123} S_{ij}v= S_{ij}v$. For instance, for $ij=12$, we have $\sigma_{123} S_{12} v = S_{12} S_{23} S_
{12} v = S_{12}  \sigma_{132} v = S_{12} v$. Analogous arguments apply to the case $ij=13,23$. This implies that each orbit contains two invariant vectors.

Defining the number of all such vectors as $\td{N}_\sigma^{(L)}$ one can count analogously as for the case of $N_{S}^L$ 
\begin{equation}\label{eq:frac_sig}
{N}_\sigma^{(L)} =  \frac{1}{2}(\td{N}_\sigma^{(L)} - N_I^{(L)})
\end{equation}
with the same notation as in Eq.~\eqref{eq:frac_s} and the factor $1/2$ coming from the fact that each orbit contains two invariant vectors. Then, $\td{N}_\sigma^{(L)}$ can then be computed as
follows. Given that a sequence of measurement outcomes is symmetric under cyclic permutation up to the step $m-1$, there are $2^{3^{m-1}}$ ways of completing it while still keeping it symmetric (see also Fig. \ref{fig9}). 

\begin{figure}[h]
\begin{center}
\includegraphics[height=0.4\columnwidth]{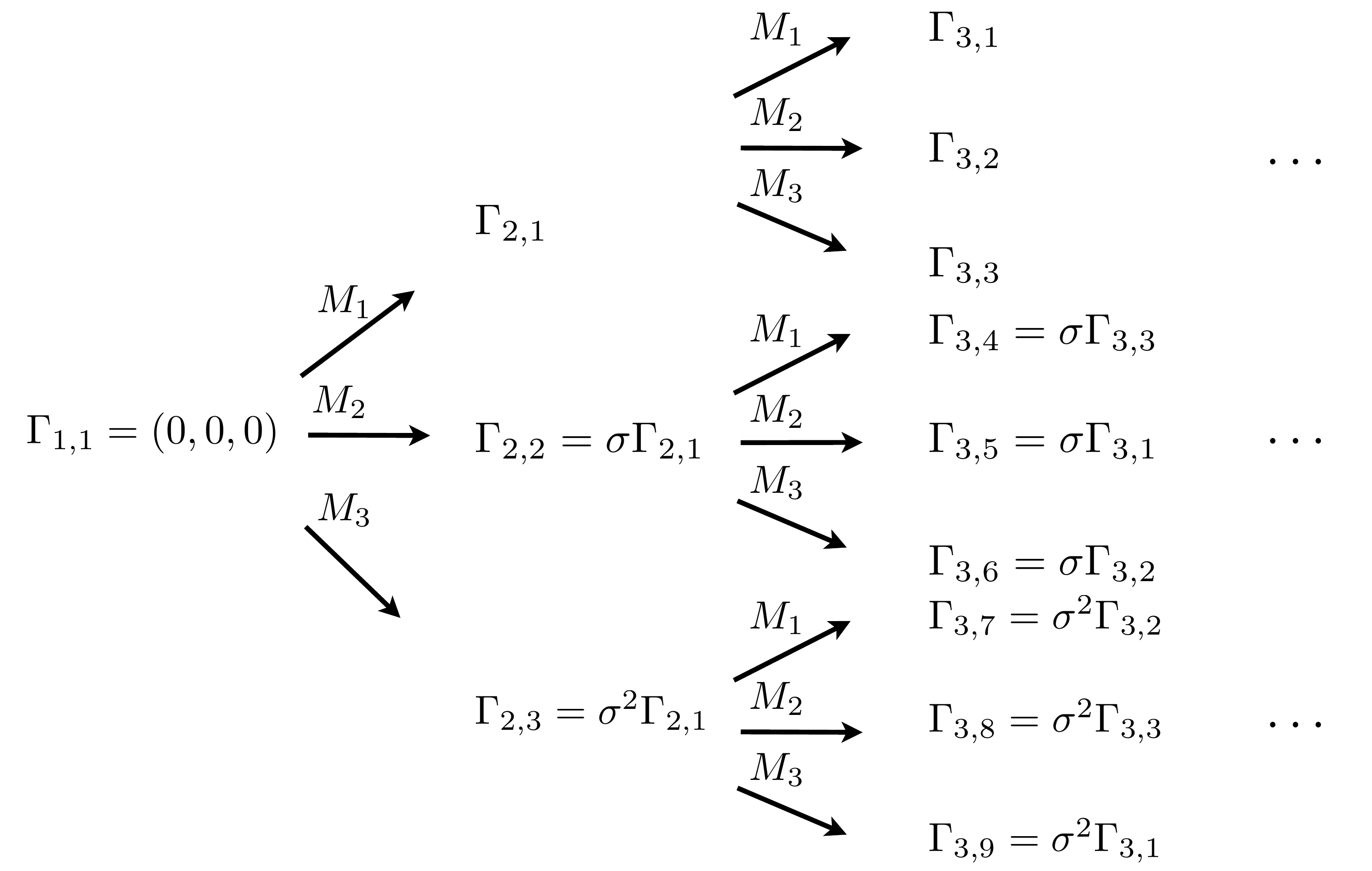} 
\end{center}
\caption{In this figure the tuples that have to be equal for an extreme point that is invariant under $\sigma_{123} $ in the scenario of three measurement settings with each two outcomes are indicated. Here $\sigma$ is an operator that permutes the outcomes of the measurements in the tuples in accordance with $M_1\rightarrow M_2\rightarrow M_3\rightarrow M_1$ and $\sigma^2$ indicates that $\sigma$ is applied twice.}
\label{fig9}
\end{figure} 

Hence, we have that \begin{equation}
\td{N}_{\sigma}^{(L)}=\prod_{i=2}^L  2^{3^{m-1}}=  2^{\frac{3^L-3}{2}}.
\end{equation}
Finally, we need to compute the number of orbits for vectors that are not invariant under any permutation. These can be obtained by removing all invariant ones from the total and divide by six, i.e., the number of vectors for each orbit, namely
\begin{equation}
\begin{split}
N_{\rm N}^{(L)} = \frac{1}{6}\left[N_{\rm ORE}^{(L)}  - (N_{\rm I}^{(L)}  + 3 N_{\rm S}^{(L)}  + 2 N_{\sigma}^{(L)} )\right]=\frac{1}{6}\left[N_{\rm ORE}^{(L)} -3\td{N}_S^{(L)}  - \td{N}_\sigma^{(L)}  + 3 N_{\rm I}^{(L)}  \right]
\end{split}
\end{equation}
Finally, we have \begin{equation}\label{eq:num_classS3}
\begin{split}
N^{(L)}  = \left[N_I^{(L)}  + (\td{N}_{S}^{(L)}  - N_{\rm I}^{(L)} ) +\frac{1}{2} ( \td{N}_\sigma^{(L)}  - N_I^{(L)} )  + \frac{1}{6}(N_{\rm ORE}^{(L)} -3\td{N}_S^{(L)}  - \td{N}_\sigma^{(L)}  + 3 N_{\rm I}^{(L)} )\right]\\
=\left[\frac{N_{\rm ORE}^{(L)} }{6} + \frac{1}{2}\td{N}_{S}^{(L)}  + \frac{2}{6} \td{N}_\sigma^{(L)}  \right]= 2^{{\frac{L-1}{2}+\frac{3(3^L-3)}{4}}-1} + \frac{1}{6}[ 2^{3\frac{3^L-3}{2}} + 2^{\frac{3^L-3}{2}+1}  ],
\end{split}
\end{equation}
wich proves the lemma.
\end{proof}

\section{Proof of Observation \ref{obs:orth_range}\label{AppendObs}}
For completeness we  prove here Observation \ref{obs:orth_range} which is a well known result in quantum state discrimination (cf., e.g., Ref.~\cite{NC_book} Ch. 9) and which is used in the main text to identify the minimal dimension of a quantum system that is required to reach a given extreme point of $P_{L}^{O,S}$. We first show the following lemma which straightforwardly extends to Observation \ref{obs:orth_range}.
\begin{lemma}\label{lemprob1}
Let $E$ be an effect of a POVM, i.e. $E\geq 0$ and $E\leq\one$, and $\rho=\sum_{i\in I} p_i\kb{\Psi_i}{\Psi_i}$ with $p_i>0$ be the spectral decomposition of a density matrix $\rho$. Then $\tr \{\rho E\}=1$ iff $E=\sum_{i\in I} \kb{\Psi_i}{\Psi_i}+\sum_{k,k'\in K} c_{k,k'}\kb{\Psi_k}{\Psi_{k'}}$ with $I\cap K=\{\}$ and $\{\ket{\Psi_i}\}_{i\in I}\cup \{\ket{\Psi_k}\}_{k\in K}$ being an ONB. The matrix $E_K =\sum_{k,k'\in K}  c_{k,k'}\kb{\Psi_k}{\Psi_{k'}}$ is positive semidefinite and $E_K\leq \one$.
\end{lemma}
\begin{proof}
\emph{If:} Inserting $E=\sum_{i\in I} \kb{\Psi_i}{\Psi_i}+\sum_{k,k'\in K} c_{k,k'}\kb{\Psi_k}{\Psi_{k'}}$ in $\tr \{\rho E\}$ and using that $\{\ket{\Psi_i}\}_{i\in I}\cup \{\ket{\Psi_k}\}_{k\in K}$ is an ONB as well as that $\sum_{i\in I} p_i=1$ readily proves the statement.\\
\emph{Only if:} Writing $E$ in the basis $\{\ket{\Psi_i}\}_{i\in I}\cup \{\ket{\Psi_k}\}_{k\in K}$, i.e. $E=\sum_{l,l'\in I\cup K}c_{l,l'}\kb{\Psi_l}{\Psi_l'}$, and inserting in $\tr \{\rho E\}=1$ one obtains that $\sum_{i\in I} c_{ii}p_i=1$.
As $0\leq E\leq \openone$ it holds that $0\leq c_{ii}\leq 1$. Moreover, using that $\sum_{i\in I} p_i=1$ and $p_i>0$ it therefore follows that $c_{ii}=1$ $\forall i \in I$. It can be easily seen that this condition and $E\leq \one$ can only be simultaneously fulfilled if $c_{ik}=0$ for $i\in I$ and $k\neq i$. More precisely, due to $E\leq \one$ it has to hold that $\bra{\Psi_i}E^\dagger E\ket{\Psi_i}=\sum_{k\in K\cup I} |c_{ki}|^2\leq 1$. As $c_{ii}=1$ for $i\in I$ we have that $c_{ki}=c_{ik}^*=0$ for $k\neq i$. Hence, $E$ is of the form $E=\sum_{i\in I} \kb{\Psi_i}{\Psi_i}+\sum_{k,k'\in K} c_{k,k'}\kb{\Psi_k}{\Psi_{k'}}$. Note that it follows immediately from $E\geq 0$ and $E\leq \one$ that $E_K =\sum_{k,k'\in K}  c_{k,k'}\kb{\Psi_k}{\Psi_{k'}}\geq 0$ and $E_K\leq \one$.
\end{proof}
It follows that states giving, with probability one, different outcomes for the same sequence of measurements have ranges corresponding to orthogonal subspaces. 
\\ \\
\noindent{\it{\bf Observation \ref{obs:orth_range}.} Let $E$ be an effect of a POVM and $\rho_1=\sum_{i\in I} p_i\kb{\Psi_i}{\Psi_i}$ with $p_i>0$ ($\rho_2=\sum_{l\in L} q_l\kb{\Phi_l}{\Phi_l}$ with $q_l>0$) the spectral decomposition of a density matrix $\rho_1$ ($\rho_2$) respectively. Then $\tr \{\rho_1 E\}=1$ and $\tr \{\rho_2 E\}=0$ only if $\braket{\Psi_i}{\Phi_l}=0$ for all $ i\in I$ and $l\in L$. }
\begin{proof}Using Lemma~\ref{lemprob1} if follows from $\tr \{\rho_1 E\}=1$ that $E=\sum_{i\in I} \kb{\Psi_i}{\Psi_i}+\sum_{k,k'\in K} c_{k,k'}\kb{\Psi_k}{\Psi_{k'}}$ with $I\cap K=\{\}$ and $\{\ket{\Psi_i}\}_{i\in I}\cup \{\ket{\Psi_k}\}_{k\in K}$ being an ONB. Denoting $\sum_{k,k'\in K} c_{k,k'}\kb{\Psi_k}{\Psi_{k'}}$ by $E_K$ we have that $\tr \{\rho_2 E\}=\sum_{i\in I}\bra{\Psi_i}\rho_2\ket{\Psi_i}+\tr \{E_K \rho_2\}=0$. As $\rho_2\geq 0$ and $E_K\geq 0$ we have that $\bra{\Psi_i}\rho_2\ket{\Psi_i}=0$ $\forall i\in I$ and it can be easily seen that this implies $\braket{\Psi_i}{\Phi_l}=0$ for all $ i\in I$ and $l\in L$.
\end{proof}

\section{A potentially improved lower bound on the dimension needed to realize any extreme point}\label{sec:appA}
In the main text we discussed a way to construct an extreme point that yields a lower bound on the necessary dimension to realize any extreme point. For any $O$ and $S$ this construction allowed to provide a closed formula for the scaling of the bound with respect to the length of the sequence. Here we  discuss a different construction which gives a potentially better lower bound.
In order to do so we consider the following extreme point. All tuples that are assigned to a time step $j<k$ correspond to $(0, 0, \ldots, 0)$, i.e. for the first $k-1$ time steps one obtains outcome ``$0$'' for all settings. In time step $t_k$ emerging subtrees $T_{k,m}$ have the property that in the root node all settings yield outcome ``$0$'', however in the second time step at least one of the tuples is not of the form $(0, 0, \ldots, 0)$. Moreover, all of these subtrees are chosen to be different, see Fig. \ref{fig10}. 
\begin{figure}
\begin{center}
\includegraphics[width=0.9\columnwidth]{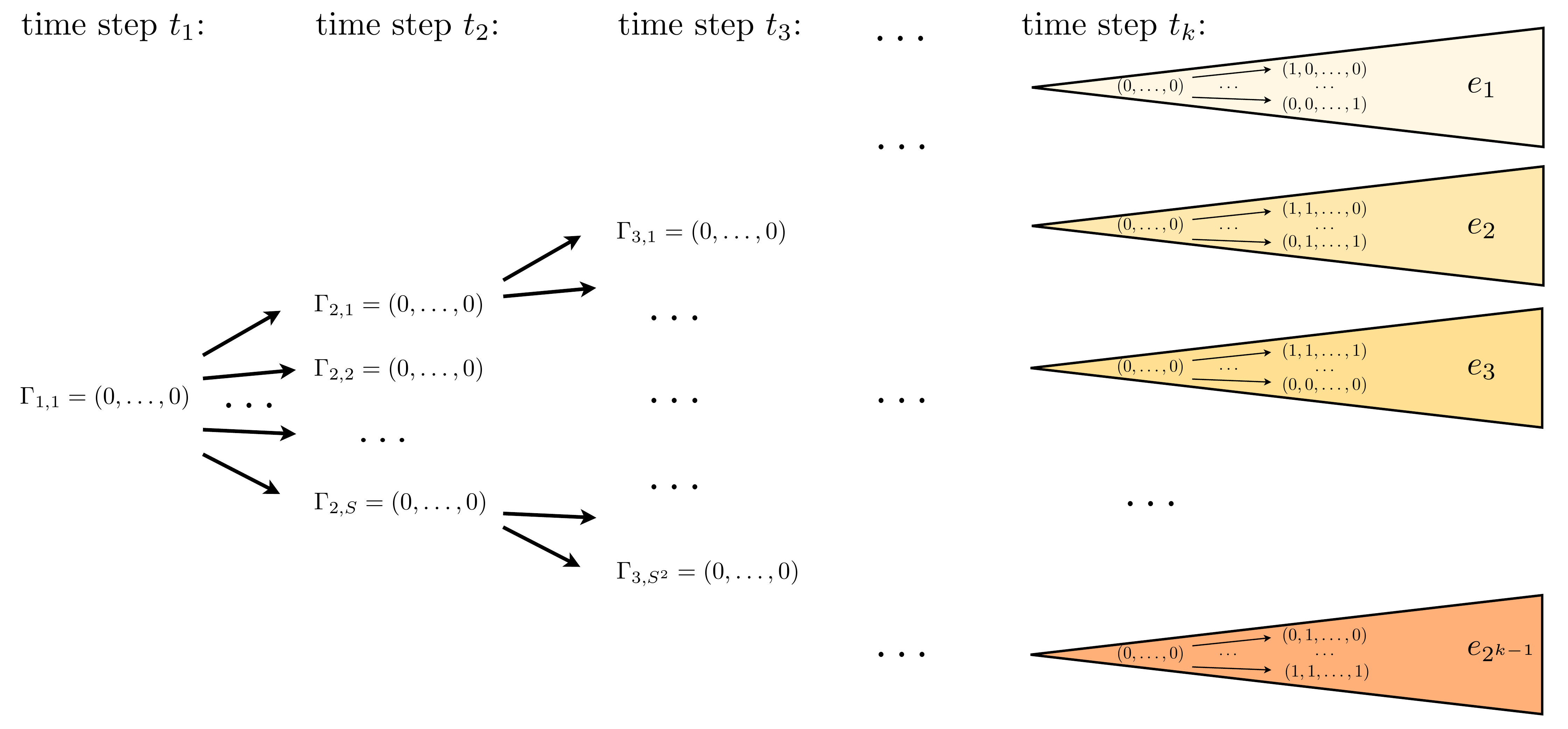} 
\end{center}
\caption{This figure illustrates the idea of the construction of the extreme point that allows to obtain a lower bound on the scaling. Here the subtrees $e_i$ have the following properties: a) In the first time step all measurements yield outcome ``$0$''. b) They are all chosen to be different. c) At least one tuple assigned to the second time step does not correspond to $(0, 0, \ldots,0)$.}
\label{fig10}
\end{figure} 

Note that therefore all possible futures assigned to a time step $i\leq k$ are not equivalent to each other. As discussed in the proof of Theorem \ref{theomindimgen} the number of inequivalent futures corresponds to the necessary dimension. Hence, one obtains straightforwardly a lower bound on the dimension given by $\sum_{i=1}^k S^{i-1}$. In order to obtain the best possible bound of this form it remains  to identify the largest $k$ for which such a construction is possible. Recall that $S^{k-1}$ is the number of futures that can be assigned to time step $t_k$, $L-k$ is the length of these futures and $(O^S)^{\frac{S^{L-k+1}-S}{S-1}}$ is the number of different futures of length $L-k$ for which the starting node is given by $(0, 0, \ldots, 0)$. The latter can be shown analogously to the proof of Lemma 1 which can be straightforwardly extended to an arbitrary number of outcomes. Note, however,  that for an arbitrary $O$ the condition that the first tuple corresponds to $(0, 0, \ldots, 0)$ does not uniquely identify one 
element of an ORE class.The number of futures of length $L-k$  for which all tuples in the second time step are equal to $(0, 0, \ldots, 0)$ is given by \begin{equation} (O^S)^{\frac{S^{L-k+1}-S^2}{S-1}}=\frac{(O^S)^{\frac{S^{L-k+1}-S}{S-1}}}{(O^S)^{S}}.\end{equation} This is due to the fact that the number of different possibilities to assign tuples in the second time step is $(O^S)^{S}$. With this it follows that one has to identify the largest natural number $k$ such that $k\leq j$ and 
\begin{equation}\label{eqAppA}
S^{j-1} \leq  (O^S)^{\frac{S^{L-j+1}-S}{S-1}}-(O^S)^{\frac{S^{L-j+1}-S^2}{S-1}}
\end{equation}
with $j\in \mathbb{R}$. Given $S$ and $L$ one can obtain $j$ for example graphically by determining the zero of $(O^S)^{\frac{S^{L-j+1}-S}{S-1}}-(O^S)^{\frac{S^{L-j+1}-S^2}{S-1}}-S^{j-1}$ which is monotonically decreasing as a function of $j$ and compute straightforwardly $k$. The lower bound on the dimension is then given by $\sum_{i=1}^k S^{i-1}=\frac{S^k-1}{S-1}$. Compared to the construction presented in the main text it may be that the $k$ obtained in the way presented here is smaller, however here all subtrees up to time step $t_k$ are accounted for and not only the ones assigned to this time step.  Hence, for certain scenario the lower bound can be improved.

\bibliographystyle{apsrev4-1}
\bibliography{Sim_Ex_TC}{}

\end{document}